\documentclass[10pt, doublecolumn]{IEEEtran}
\usepackage{cite}
\ifCLASSINFOpdf
  \usepackage[pdftex]{graphicx}
  \usepackage{graphicx}
  \graphicspath{{../pics/}}
\else
\fi

\usepackage{amsmath}
\usepackage{amsthm}
\usepackage{amssymb}
\usepackage{mathrsfs}

\usepackage{algorithm}  
\usepackage{algorithmicx}  
\usepackage{algpseudocode}  

\usepackage{array}  
\usepackage{subfigure}

\usepackage[utf8]{inputenc}
\usepackage[english]{babel}
\usepackage[usenames,dvipsnames]{color}

\usepackage[justification=centering]{caption}

\usepackage{stfloats}

\usepackage{color}
\usepackage{lipsum}
\usepackage{cuted}

\usepackage{booktabs}
\usepackage{makecell}

\usepackage{cuted}

\newtheorem{lemma}{Lemma}
\newtheorem{corollary}{Corollary}

\begin{document}
%
\title{Reconfigurable Intelligent Surface Deployment for Wideband Millimeter Wave Systems}
%
%
%

\author{Xiaohao~Mo,~
        Lin~Gui,~\IEEEmembership{Member,~IEEE,}
        Kai~Ying,~\IEEEmembership{Senior Member,~IEEE,}
        Xichao~Sang,~
        and~Xiaqing~Diao~
\thanks{X. Mo, L. Gui, K. Ying, X. Sang and X. Diao are with the Department of Electronic Engineering, Shanghai Jiao Tong University, Shanghai 200240, China (email: mxh\_sjtu@sjtu.edu.cn; guilin@sjtu.edu.cn; yingkai0301@sjtu.edu.cn; sang\_xc@sjtu.edu.cn; xqdiao99@sjtu.edu.cn).}
\thanks{L. Gui is the corresponding author.}}

\maketitle

\begin{abstract}
The performance of wireless communication systems is fundamentally constrained by random and uncontrollable wireless channels. Recently, reconfigurable intelligent surfaces (RIS) has emerged as a promising solution to enhance wireless network performance by smartly reconfiguring the radio propagation environment. While significant research has been conducted on RIS-assisted wireless systems, this paper focuses specifically on the deployment of RIS in a wideband millimeter wave (mmWave) multiple-input-multiple-output (MIMO) system to achieve maximum sum-rate. First, we derive the average user rate as well as the lower bound rate when the covariance of the channel follows the Wishart distribution. Based on the lower bound of users' rate, we propose a heuristic method that transforms the problem of optimizing the RIS's orientation into maximizing the number of users served by the RIS. Simulation results show that the proposed RIS deployment strategy can effectively improve the sum-rate. Furthermore, the performance of the proposed RIS deployment algorithm is only approximately 7.6\% lower on average than that of the exhaustive search algorithm.

\end{abstract}


%
\IEEEpeerreviewmaketitle

\section{Introduction}
\IEEEPARstart{O}{wing} to the capability of dynamically reconfiguring the radio propagation environment and its characteristic of low cost, low power consumption and  programmability and easy deployment, the reconfigurable intelligent surface (RIS) can improve the performance of coverage enhancement, interference suppression, physical layer security, high-precision positioning, and other aspects by interacting with user devices, thereby further enhancing the performance and resource utilization efficiency of communication systems, meeting the requirements of next-generation mobile communication systems\cite{basar2019wireless,di2020smart,huang2019reconfigurable}.
With the help of RIS, a  smart radio environment (SRE) \cite{renzo2019smart,wu2019towards} can be built to meet the needs of the future mobile communication system.

Recently, researchers have been enthusiastically fostering the RIS for wireless communication improvement, the interesting topics are not limited to physical prototype \cite{arun2020rfocus}, channel estimation \cite{wei2021channel}, joint precoding/beamforming \cite{wu2019intelligent,huang2018achievable,zhang2021joint,yang2019irs}, etc. However, the correct deployment of RIS is also a significantly crucial research subject. The positioning of RIS directly impacts the quality of the BS-RIS link and the RIS-UE link, thereby determining the users' Quality of Service (QoS). Currently, research on deployment strategies of RIS remains an open area of investigation.

The three-dimensional (3D) ray-tracing channel simulation of RIS is studied in \cite{Huang_Novel_2022}, and the proposed model in verified in multiple aspects, such as received power, channel capacity, and angular power spectral density (PSD). In addition, the impact of different deployment positions of the RIS on coverage range is evaluated, which indirectly emphasizes the importance of correctly configuring the RIS' location.
Authors in \cite{deb2021ris} and \cite{ghatak2021placement} utilize the characteristic of the RIS that can create virtual line-of-sight (LoS) to {overcome} obstacles between transceivers. Specifically, authors in \cite{deb2021ris} use a graph theory-based method to minimize the number of RISs, ensuring that there exists at least a direct LoS or virtual LoS between device pairs for the millimeter wave (mmWave) device to device (D2D) system. In \cite{ghatak2021placement}, a RIS deployment strategy is proposed to minimize the user blockage probability based on the given obstacle distribution in a single-user single-input-single-output (SISO) scenario.
{Moreover}, authors in \cite{zeng2020reconfigurable,nemati2020ris,cui2021snr} intend to increase the coverage of the cell to serve more users by deploying the RIS. In the single-user SISO scenario, authors in \cite{zeng2020reconfigurable} take the orientation of the RIS into consideration, and propose a coverage maximization algorithm to determine the orientation and location of the RIS. In \cite{nemati2020ris} and \cite{cui2021snr}, optimal RIS deployment strategies are proposed to improve the signal-to-interference ratio (SIR) and signal-to-noise ratio (SNR) of edge users of the cell, respectively.

Another deployment strategy is to design {with} the purpose of increasing the average/sum rate of users in the cell \cite{bie2022deployment,xu2021reconfigurable,tian2022optimizing,ntontin2020reconfigurable,chen2022reconfigurable,zhang2022reconfigurable,liu2020ris}. Both authors in \cite{bie2022deployment} and \cite{xu2021reconfigurable} compare the performance of the RIS-based system and the relay-based system in the end-to-end SISO scenario. The upper bound capacity of both RIS and relay systems are deduced in \cite{bie2022deployment}, and the optimal location of the RIS is also determined. In \cite{xu2021reconfigurable}, the 3D path loss model is introduced, and the optimal number of phase shifters of the RIS is optimized. {Additionally}, in \cite{tian2022optimizing}, the area average rate is maximized by optimizing the location, height and downtilt of the RIS in the mmWave vehicular system. In the highly directional mmWave links, the approximate closed-form expression of received signal power is provided in \cite{ntontin2020reconfigurable}, and the location of RIS is optimized by maximizing the end-to-end SNR. In the RIS-assisted non-orthogonal multiple access (NOMA) system, a machine-learning (ML) based method is proposed to determine the optimal location of the RIS \cite{liu2020ris}. Specifically, authors firstly propose a novel long short-term memory (LSTM) to predict users' traffic demand and then utilize the decaying double deep Q-network ($\text{D}^\text{3} \text{QN}$) to solve the deployment problem. Furthermore, authors in \cite{Qin_Indoor_2022} investigate the deployment strategy of RIS in indoor scenarios, aiming to minimize the interruption probability considering the distribution and mobility model of human bodies. The specific location of the RIS is determined by exhaustive search method.

However, the existing literature on communication models in wireless systems tends to be relatively simple, with many of them only considering the SISO scenario, and effective optimization algorithms for RIS deployment have yet to be developed. Furthermore, these models do not typically account for the long-term geographic distribution of users, since the location of the RIS is usually fixed and cannot be casually moved. In this paper, we focus on the wideband RIS-assisted mmWave multiple-input-multiple-output (MIMO) system, and aim to establish an optimization problem for determining the optimal location of the RIS based on the long-term geographic distribution of users, and propose a heuristic method for solving it. The main contributions of this paper are summarized as follows:

\begin{itemize}
  \item This paper investigates the optimization problem of the deployment of the RIS in a wideband mmWave MIMO system, with the objectives of maximizing the sum-rate of users in a cell. In addition to optimizing the spatial coordinates of the RIS, we also consider the orientation of the RIS. Furthermore, taking into account the variations in location of users and the fixed position of the RIS, a new deployment strategy based on the long-term geographic distribution of users is proposed.
  \item Due to the high coupling between the location of the RIS and the channel parameters, we firstly derive the average user rate and its lower bound by exploiting the characteristics of the channel covariance matrix following the Wishart distribution. A heuristic algorithm based on the lower bound rate is proposed, and the sample average approximation (SAA) method is used to handle the expectation computation in the objective function. Specifically, to the best of our knowledge, the optimization of the orientation of the RIS cannot be solved directly due to its tight coupling with channel parameters. Therefore, we transform the optimization problem to maximize the number of users served by the RIS. For the remaining location parameters, we find the optimal solution by analyzing the derivative of the lower bound rate. At last, the complexity and convergence analysis of the proposed algorithm are presented.
  \item The simulation results present comparisons of RIS deployment strategies in different geographic distributions of users. Our findings show that the sum-rate based on the proposed heuristic method outperforms that of the stochastic gradient descent (SGD) algorithm and is comparable to that of the exhaustive search method. We also investigate the effects of different RIS location parameters on system performance through numerical simulations.
\end{itemize}

Finally, we compare the proposed scheme with existing schemes in Table \ref{table-compare}.
\begin{table}[hpt]
  \caption{Compare with existing schemes}
  \label{table-compare}
  \centering
  \begin{tabular}{|c|c|c|c|c|c|c|}
    \hline
         & \cite{zeng2020reconfigurable} & \cite{nemati2020ris} &  \cite{bie2022deployment} & \cite{tian2022optimizing} & \cite{ntontin2020reconfigurable} & Proposed \\
    \hline
    Single User  &   $\checkmark$   &     & $\checkmark$    & $\checkmark$    & $\checkmark$     &     \\
    \hline
    Multiple users   &       &  $\checkmark$    &     &       &      &  $\checkmark$  \\
    \hline
    Orientation of the RIS   &  $\checkmark$    &     &    & $\checkmark$     &     & $\checkmark$   \\
    \hline
    Fixed location of users  &  $\checkmark$  & $\checkmark$    & $\checkmark$   &     &$\checkmark$     & \\
    \hline    
    {\makecell[c]{Long-term geographic  \\distribution of users}}&      &     &    &     &     & 
     $\checkmark$    \\
     \hline
     Capacity   &   &  & $\checkmark$   & $\checkmark$     & $\checkmark$     & $\checkmark$    \\
     \hline
     Coverag  & $\checkmark$     & $\checkmark$    &    & $\checkmark$     &     & \\
     \hline
     mmWave  &  &  &    & $\checkmark$     & $\checkmark$     & $\checkmark$  \\
     \hline
     Sub-6 G  & $\checkmark$     &$\checkmark$    & $\checkmark$   &    &     & \\
     \hline    
     Wideband &  & &  &  &  &  $\checkmark$ \\
     \hline      
  \end{tabular}
\end{table}

The rest of this paper is organized as follows. Section II introduces the system model, the channel model and the RIS deployment optimization problem formulation. In order to solve the optimization problem of RIS' location efficiently, the average user rate expression is derived in Section III. In Section IV, the heuristic RIS deployment strategy is proposed as well as the analysis of computation complexity and convergence, and the optimization of RIS' phase shifters is also presented. The simulation results are presented in Section V. Finally, the conclusion and future works are discussed in section VI.

\section{System Model}
\begin{figure}[hbtp]
  \centering
  \includegraphics[width = 80mm]{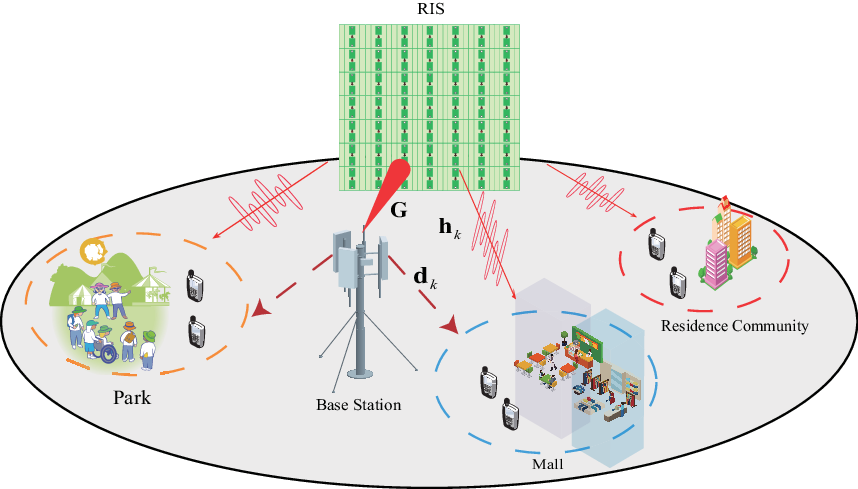}
  \caption{RIS-assisted mmWave cellular system.}
  \label{fig-systemmodel} 
\end{figure}
As shown in Fig. \ref{fig-systemmodel}, we consider a wideband RIS-assisted mmWave MIMO system, where a BS and an RIS serve $K$ single-antenna users simultaneously by $M$ subcarriers. Specifically, the BS is equipped with $N_{t}$-element uniform linear array (ULA) antennas and the RIS is equipped with $N_{r}$-element uniform planar array (UPA) phase shifters. We assume that the RIS can be controlled by the BS through a wired or wireless link. The users are distributed in various locations, like communities, shopping malls, parks, etc. Besides, key parameters are listed in Table \ref{key-para}.

\begin{table*}[hpt]
  \caption{Definitions of significant parameters.}
  \label{key-para}
  \centering
  \begin{tabular}{|c|l|}  
    \hline
    \textbf{Parameters}      & \makecell[c]{\textbf{Definitions}}  \\ 
    \hline
    $K / M / N_{t} / N_{r}$  & The number of users$/$subcarriers$/$BS antennas$/$RIS elements \\
    \hline
    $\mathcal{K}/\mathcal{M}/\mathcal{N}{t}/\mathcal{N}_{r}$  & The set of users$/$subcarriers$/$BS antennas$/$RIS elements \\
    \hline
    $\bar{\mathbf{G}}_{m}/\bar{\mathbf{d}}_{k,m}/\bar{\mathbf{h}}_{k,m}$  & The LoS channel component of BS-RIS link$/$BS-user $k$ link$/$RIS-user $k$ link at the $m$-th subcarrier \\
    \hline
    $\tilde{\mathbf{G}}_{m}/\tilde{\mathbf{d}}_{k,m}/\tilde{\mathbf{h}}_{k,m}$  & The NLoS channel component of BS-RIS link$/$BS-user $k$ link$/$RIS-user $k$ link at the $m$-th subcarrier \\
    \hline
    $\beta_{0} / \beta_{1,k} / \beta_{2,k} $ & The pathloss of BS-RIS link$/$BS-user $k$ link$/$RIS-user $k$ link \\
    \hline
    $K_{0} / K_{1} / K_{2} $ & The Rician factor of BS-RIS link$/$BS-user $k$ link$/$RIS-user $k$ link \\
    \hline
    $\omega_{k}\in \{0,1\}$ & The indicator of whether the user $k$ is served by the RIS \\
    \hline
    $ \mathbf{p} = \left(d_0,\varphi_0,h_{0}\right) / \varphi_{R}$ & The 3D coordinate of the RIS$/$The orientation of the RIS\\
    \hline
     $ \mathbf{q}_{k} = \left(d_k,\varphi_k,h_{u}\right)$ &  The 3D coordinate of user $k$\\
    \hline
  \end{tabular}
\end{table*}

\subsection{Channel Model}

Let $\mathbf{d}_{k,m} \in \mathbb{C}^{N_{t} \times 1}$ and $\mathbf{h}_{k,m} \in \mathbb{C}^{N_{r} \times 1}$ denote channel from the BS to user $k$ and from the RIS to user $k$ at the $m$-th subcarrier. $\mathbf{G}_{m} \in \mathbb{C}^{N_{t} \times N_{r}}$ denotes the channel between the BS and RIS at $m$-th subcarrier. According to \cite{wang_Pervasive_2022}, channels can be divided into two parts, which are LoS and non-line-of-sight (NLoS) component. Then, in this paper, channels can be further expressed as 
\begin{align}
  \left\{\begin{array}{ll}
  \mathbf{G}_{m} &= \sqrt{\beta_{0}\frac{K_{0}}{K_{0}+1}} \bar{\mathbf{G}}_{m} + \sqrt{\beta_{0}\frac{1}{K_{0}+1}} \tilde{\mathbf{G}}_{m}, \\
  \mathbf{d}_{k,m} &= \sqrt{\beta_{1,k}\frac{K_{1}}{K_{1}+1}}\bar{\mathbf{d}}_{k,m} + \sqrt{\beta_{1,k}\frac{1}{K_{1}+1}}\tilde{\mathbf{d}}_{k,m}, \\
  \mathbf{h}_{k,m} &= \sqrt{\beta_{2,k}\frac{K_{2}}{K_{2}+1}}\bar{\mathbf{h}}_{k,m} + \sqrt{\beta_{2,k}\frac{1}{K_{2}+1}}\tilde{\mathbf{h}}_{k,m},
  \end{array} \right.
\end{align} 
where $\bar{\mathbf{G}}_{m}$, $\bar{\mathbf{d}}_{k,m}$ and $\bar{\mathbf{h}}_{k,m}$ are the LoS component of $\mathbf{G}_{m}$, $\mathbf{d}_{k,m}$ and $\mathbf{h}_{k,m}$, respectively, while $\tilde{\mathbf{G}}_{m}$, $\tilde{\mathbf{d}}_{k,m}$ and $\tilde{\mathbf{h}}_{k,m}$ are the NLoS component. $\beta_{0}$, $\beta_{1,k}$ and $\beta_{2,k}$ represent the large scale path loss of the BS-RIS link, BS-user $k$ link and RIS-user $k$ link, respectively. $K_{i} \left(i=0,1,2\right)$ is the Rician factor of the corresponding channels. Based on \cite{Overview_Heath_2016} and \cite{Millimeter_Rappaport_2013}, the LoS channel components can be respectively expressed as
\begin{align}
  \left\{\begin{array}{ll}
  \bar{\mathbf{G}}_{m} &= \mathbf{b}^{H}\left( \phi_{0,m} \right) \mathbf{a}\left( \theta_{0,m}^{a}, \theta_{0,m}^{e}\right), \\
  \bar{\mathbf{d}}_{k,m} &= \mathbf{b}^{H}\left( \phi_{1,k,m} \right), \\
  \bar{\mathbf{h}}_{k,m} &= \mathbf{a}^{H}\left( \theta_{2,k,m}^{a}, \theta_{2,k,m}^{e}\right),
  \end{array}\right.
\end{align}
where $\phi_{0,m}$ denotes the spatial direction at the $m$-th subcarrier, and $\theta_{0,m}^{a}$ ($\theta_{0,m}^{e}$) is the azimuth (elevation) spatial direction at the $m$-th subcarrier of the BS-RIS link. $\phi_{1,k,m}$ is the spatial direction at BS on the $m$-th subcarrier of the BS-user $k$ link, and $\theta_{2,k,m}^{a}$ ($\theta_{2,k,m}^{e}$) is the azimuth (elevation) spatial direction at RIS on the $m$-th subcarrier of the RIS-user $k$ link. And $\mathbf{b}\left( \cdot \right)$ and  $\mathbf{a}\left( \cdot \right)$ are the array response vectors for ULA and UPA, respectively. The detailed expressions are shown as
\begin{align}
  \mathbf{b}_{N_{t}}\left(\vartheta\right) &= \frac{1}{\sqrt{N_{t}}} \left[1, \ldots, e^{j 2 \pi  \vartheta},\ldots,e^{j 2 \pi (N_{t}-1) \vartheta}\right]^{T},  \\
  \mathbf{a}_{N_{r}}\left(\vartheta^{a}, \vartheta^{e}\right)&=\frac{1}{\sqrt{N_{r}}}\left[1, \ldots, e^{j 2 \pi  \vartheta^{a}},\ldots,e^{j 2 \pi (N_{r}^{x}-1) \vartheta^{a}}\right]^{T}  \notag \\
  &\otimes \left[1, \ldots, e^{j 2 \pi  \vartheta^{e}},\ldots,e^{j 2 \pi (N_{r}^{y}-1) \vartheta^{e}}\right]^{T}, 
\end{align}
where $N_{r} = N_{r}^{x} \times N_{r}^{y}$. 
Due to the wideband effect \cite{Beamspace_Mo_2021,Beam_wang_2019}, the spatial direction is defined as the following formula, which is different among different subcarriers.
\begin{align}
  \left\{\begin{array}{l} 
    {\phi}_{0,m} = \frac{f_{m}}{c}d\sin\left( \tilde{\phi}_{0} \right), \\
    \theta_{0,m}^{j} = \frac{f_{m}}{c}d\sin\left( \tilde{\theta}_{0}^{j} \right), j = \{\mathrm{a},\mathrm{e}\}, \\
    {\phi}_{1,k,m} = \frac{f_{m}}{c}d\sin\left( \tilde{\phi}_{1,k} \right), k \in \mathcal{K}, \\
    \theta_{2,k,m}^{j} = \frac{f_{m}}{c}d\sin\left( \tilde{\theta}_{2,k}^{j} \right), j = \{\mathrm{a},\mathrm{e}\},k \in \mathcal{K}, \\
    \end{array}\right.
\end{align}
where $f_{m}=f_{c}+\frac{B}{M}\left(m-1-\frac{M-1}{2}\right)$ is the frequency of the $m$-th subcarrier with $f_c$ and $B$ representing center carrier frequency and the bandwidth, $c$ is the speed of light, and $d$ is the antenna spacing. $\tilde{\phi}_{0}$ is the angle of departure (AoD) from the BS towards the RIS, and $\tilde{\theta}_{0}^{a}$ ($\tilde{\theta}_{0}^{e}$) is the azimuth (elevation) angle of arrival (AoA) at the RIS from the BS. $\tilde{\phi}_{1,k}$ is the AoD from the BS towards user $k$, and $\tilde{\theta}_{2,k}^{a}$ ($\tilde{\theta}_{2,k}^{e}$) is the azimuth (elevation) AoD from the RIS to user $k$.

As for the NLoS channel component, we assume $\tilde{\mathbf{G}}_{m}$, $\tilde{\mathbf{d}}_{k,m}$ and $\tilde{\mathbf{h}}_{k,m}$ are independently and identically distributed (i.i.d) complex Gaussian random variables, whose elements are following the distribution of $\mathcal{CN}\left(0,1\right)$.

Then, the overall channel from the BS to user $k$ at the $m$-th subcarrier can be written as 
\begin{align}
  {\mathbf{h}}_{k,m}^{\text{eff}} = \mathbf{d}_{k,m} + \omega_{k} \mathbf{G}_{m} \boldsymbol{\Phi} \mathbf{h}_{k,m},
\end{align}
where $\boldsymbol{\Phi} = \text{diag} \left( \boldsymbol{\theta} \right) \in \mathbb{C}^{N_{r} \times N_{r}}$ is phase shifter matrix of RIS with constant constraint, i.e., $ | \theta_{i} | = 1, \forall i \in \mathcal{N}_{r}$. Besides, since the RIS can only enhance communication on one side, $\omega_{k} \in \{0,1\}$ denotes whether the user $k$ is served by the RIS. 

\subsection{RIS Deployment}
\begin{figure}[hbtp]
  \centering
  \includegraphics[width = 80mm]{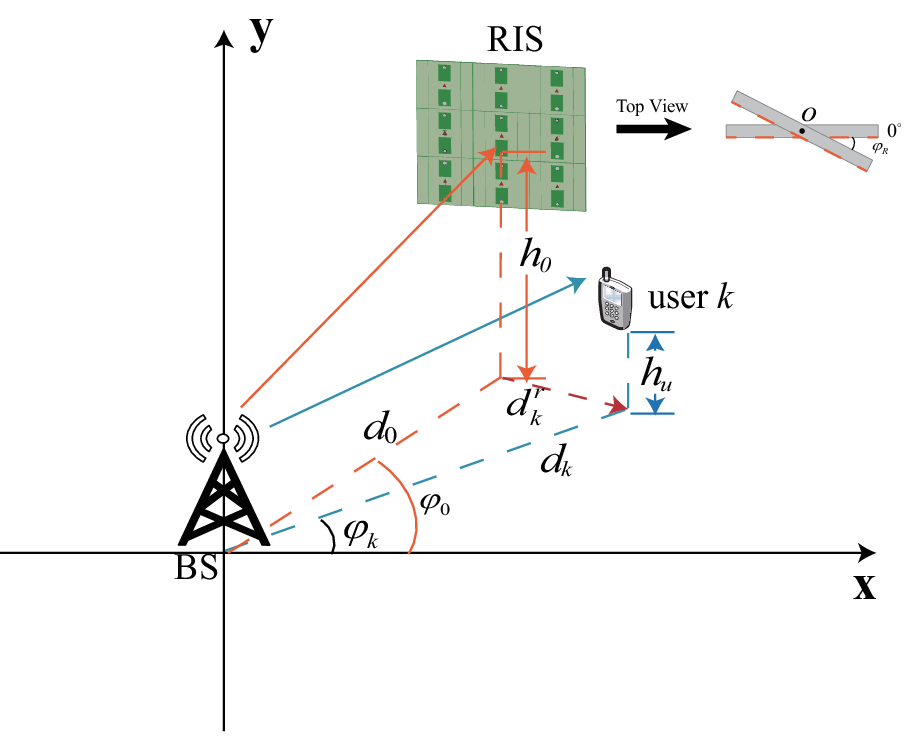}
  \caption{The geometric description of the system.}
  \label{fig-topview}
\end{figure}

Fig. \ref{fig-topview} shows the geometric description of the RIS-assisted mmWave cellular system, where the BS is located at the center of the cell, and the radius of a cell is {denoted} as $r$. The coordinates of the BS is $ \left(0,0,h_{B}\right)$, where the $h_{B}$ denotes the height of the BS. Let $ \mathbf{p} = \left(d_0,\varphi_0,h_{0}\right)$ and $\mathbf{q}_{k} = \left(d_k,\varphi_k,h_{u}\right)$ denote the location of RIS and user $k$, respectively. $d_{0}$ ($d_{k}$) is the horizontal distance from the BS towards the RIS (user $k$), $\varphi_{0}$ ($\varphi_{k}$) is the azimuth angle between the BS-RIS (user $k$) link and the positive of x axis, and $h_{0}$ ($h_{u}$) is the height of the RIS (all users). Moreover, we suppose that the orientation of the RIS is also a significant parameter to improve the system performance. 
$\varphi_{R}$ is defined as direction of the RIS' orientation between the RIS and the x-axis with counter clock-wise set to be positive. In this paper, we do not consider the elevation angle of the RIS, since the RIS is basically vertical to the ground. Moreover, we assume that the location of RIS cannot be adjusted casually, which should be optimized based on the users' long-term geographical distribution.

Based on the coordinates of the BS and the RIS, the AoD $\tilde{\phi}_{0}$ of the LoS channel component at BS can be expressed as $\varphi_{0}$, and the AoA at the RIS from the BS is 
\begin{align}
  \left\{\begin{array}{ll}
  \tilde{\theta}_{0}^{a} &= \frac{\pi}{2} - \varphi_{0} - \varphi_{R},  \\
  \tilde{\theta}_{0}^{e} &= \arctan \frac{ | h_{B} - h_{0} |}{d_{0}} .
  \end{array}\right.
\end{align} 

As for the LoS channel component of the RIS-user link, the distance between the RIS and user $k$ can be obtained by law of cosines, which is 
\begin{align} \label{dk_r}
  d_{k}^{r} = \sqrt{d_{0}^2 + d_{k}^2 - 2d_{0}d_{k}\cos(\varphi_{0} - \varphi_{k})}.
\end{align}
And the corresponding AoD can be denoted by 
\begin{align}
  \left\{\begin{array}{ll}
  \tilde{\theta}_{2,k}^{a} &= \arccos \frac{d_{0}^2 + (d_{k}^{r})^2 - d_{k}^{2}  }{2d_{0}d_{k}^{r}} - \left( \frac{\pi}{2} - \varphi_{0}  \right)  - \varphi_{R},  k \in \mathcal{K},\\
  \tilde{\theta}_{2,k}^{e} &= \arctan \frac{ | h_{u} - h_{0} |}{d_{k}^{r}}, k \in \mathcal{K}. 
  \end{array}\right.
\end{align}
 
When the BS and users are in the same side of the RIS, we suppose that users can be served by the RIS, i.e.,
\begin{equation}
  \left\{\begin{array}{l} 
    \omega_{k} = 1, \text{user $k$ in } \mathcal{S}_{R},\\
    \omega_{k} = 0, \text{otherwise},
    \end{array}\right.
\end{equation} 
where $\mathcal{S}_{R}$ denotes served area of the RIS.Obviously, the location of the RIS can not only affect the number of users that can be served by the RIS, but also the quality of the BS-RIS-user links. Therefore, it is important to optimize the location of the RIS.

\subsection{Problem Formulation}
With the help of the RIS, the downlink signal at user $k$ can be expressed as
\begin{align}
  y_{k,m} = \left({\mathbf{h}}_{k,m}^{\text{eff}} \right)^{H}\mathbf{x}_{m} + n_{k,m},
\end{align}
where $\mathbf{x}_{m} \in \mathbb{C}^{N_{t} \times 1}$ is the precoded transmitted signal, and $n_{k,m}$ denotes the white Gaussian noise following the distribution of $\mathcal{CN}\left(0,1\right)$. Specifically, the $\mathbf{x}_{m}$ can be further expressed as $\mathbf{x}_{m}  = \sum_{k=1}^{K}\sqrt{p_{k,m}}\mathbf{f}_{k,m}s_{k,m}$, where $p_{k,m}$ is the transmitted power for user $k$ at the $m$-th subcarrier, $\mathbf{f}_{k,m}$ is the precoding vector with normalized power and $s_{k,m}$ is the transmitted symbol to user $k$ at the $m$-th subcarrier. The sum of transmitted power at the BS is limited by $P_{T}$, which is $P_{T} = \sum_{m=1}^{M}\sum_{k=1}^{K}p_{k,m}$. Let $\mathbf{s}_{m} \triangleq\left[s_{1,m}, \cdots, s_{K,m}\right]^{T} \in \mathbb{C}^{K}$, and we suppose that the transmitted symbols have normalized power, i.e., $\mathbb{E}\left[\mathbf{s}_{m}\mathbf{s}_{m}^{H} \right]= \mathbf{I}_{K}$, $\forall m \in \mathcal{M}$.

Then, the signal to interference plus noise ratio (SINR) for user $k$ at the $m$-th subcarrier can be obtained as 
\begin{align}
  \mathrm{SINR}_{k.m} = \frac{p_{k,m} | \left( {\mathbf{h}}_{k,m}^{\text{eff}} \right)^{H} \mathbf{f}_{k,m}|^{2} }{\sum_{i \neq k}^{K}p_{i,m}|  \left( {\mathbf{h}}_{k,m}^{\text{eff}} \right)^{H} \mathbf{f}_{i,m}|^{2}  + \sigma_{k,m}^{2} },
\end{align}
where $\sigma_{k,m}^{2} = | n_{k,m}^{H}n_{k,m} |$ denotes the noise power, and we assume that $\sigma_{i,j}^{2} = \sigma_{\hat{i},\hat{j}}^{2} = \sigma_{2}$ for $\left(i,\hat{i}\right) \in \mathcal{K}$ and $\left(j,\hat{j}\right) \in \mathcal{M}$. 
To mitigate multi-user interference within the cell, we adopt zero forcing (ZF) precoder as in \cite{Xiao_Average_2022}.
The ZF precoder can be expressed as 
\begin{align}
  \mathbf{U}_{m} = \left({\mathbf{H}}_{m}^{\text{eff}}\right)^{H}\left( {\mathbf{H}}_{m}^{\text{eff}}\left({\mathbf{H}}_{m}^{\text{eff}}\right)^{H}\right)^{-1},
\end{align}
where $\left({\mathbf{H}}_{m}^{\text{eff}}\right) = \left[\left({\mathbf{h}}_{1,m}^{\text{eff}} \right)^{H};\cdots; \left({\mathbf{h}}_{K,m}^{\text{eff}} \right)^{H}\right] \in \mathbb{C}^{K \times N_{t}}$. As for the acquisition of CSI, the main idea is to exploit the sparsity of the mmWave channel or the RIS' cascaded channel and design efficient algorithms based on compressive sensing theory \cite{wei2021channel,Mo_Wideband_2022a}. It should be noted that we use instantaneous channel state information (CSI) for ZF precoding in this paper, and the ZF precoder $\mathbf{f}_{k,m}$ can be obtained by $ \frac{ \mathbf{u}_{k,m}}{\| \mathbf{u}_{k,m}\|}$ where $\mathbf{u}_{k,m}$ is the $k$-th column of $\mathbf{U}_{m}$. Due to the multi-user interference elimination by $\mathbf{f}_{k,m}$, the instantaneous rate of  user $k$ at the $m$-th subcarrier is
\begin{align}
  R_{k,m} = \log_{2}\left(1 +\frac{p_{k,m}  }{ \sigma^{2} \| \mathbf{u}_{k,m}\|^{2} }\right), 
\end{align}
where $\| \mathbf{u}_{k,m}\|^{2}$ can be obtained by $\left[\left( {\mathbf{H}}_{m}^{\text{eff}}\left({\mathbf{H}}_{m}^{\text{eff}}\right)^{H}\right)^{-1}\right]_{k,k}$. Moreover, due to the presence of random components in the channel, the average rate of user $k$ at the $m$-th subcarrier can be expressed as 
\begin{align} \label{R_km}
  R_{k,m}\left( \mathbf{p}, \varphi_{R}, \mathbf{q}_{k} \right) = {\mathbb{E}} \left[\log_{2}\left(1 +\frac{p_{k,m}  }{ \sigma^{2} \| \mathbf{u}_{k,m}\|^{2} }\right) \right],
\end{align}
where $\mathbb{E}$ is the expectation operator. Then, the sum-rate of the cell can be expressed as 
\begin{align}
  R \left( \mathbf{p},\varphi_{R}, \mathbf{Q} \right)  = \sum_{k=1}^{K} \sum_{m = 1}^{M} R_{k,m} \left( \mathbf{p}, \varphi_{R},\mathbf{q}_{k} \right) .
\end{align}
where $\mathbf{Q} = \{ \mathbf{q}_{1},\mathbf{q}_{2},\cdots,\mathbf{q}_{K}  \}$ denotes the locations of $K$ users.

In this paper, we want to find the optimal RIS' position to maximize the sum-rate. Due to the fixity of the RIS' location and the variations in location of users, we want to maximize the sum rate based on the long term users' geographical distribution by optimizing the location of RIS $\mathbf{p}$ and the orientation $\varphi_{R}$. Finally, the optimization problem is formulated as 
\addtocounter{equation}{0}
\begin{IEEEeqnarray}{lcll} 
  \mathcal{P} 1:  &\max _{ \mathbf{p},\varphi_{R}} \quad & \mathbb{E}_{\mathbf{Q}} \left[  {R}\left( \mathbf{p},\varphi_{R},\mathbf{Q}  \right)\right] \label{p1} \\
  &\text{s.t.}  & \mathbf{p} \in \mathcal{S},     \IEEEyessubnumber\label{p1_constraint1}\\
  &             & \varphi_{R} \in [ 0,2\pi ),    \IEEEyessubnumber\label{P1_constraint2}
\end{IEEEeqnarray}
where $\mathcal{S}$ denotes the area of the cell. (\ref{p1_constraint1}) implies that the RIS should be within the cell range. The objective function in $\mathcal{P} 1$ refers to the average achievable rate achieved by averaging throughout all the user locations, which follows a specific distribution.

Obviously, solving (\ref{p1}) is challenging for several reasons. Firstly, obtaining a definite closed-form expression of the objective function is difficult due to the expectation operator and the matrix inversion calculation. Secondly, the optimization variables are intricately coupled with the system parameters. Thirdly, the objective function is non-convex, which makes the problem intractable. For the stochastic optimization problem, an intuitive solution is to adopt the SGD-based method.
However, the gradient of location parameters cannot be obtained easily since the objective function involves matrix inversion. 
Inspired by \cite{Xie_Performance_2023,Li_Achievable_2023} and \cite{Singh_Performance_2023}, we firstly give the analytical expressions of the performance in Section \ref{sec_appro_rate}, which refer to the approximation of average user rate of $R_{k,m}$ as well as its lower bound, and then propose a heuristic algorithm in section IV-A.

\section{Average Rate approximation} \label{sec_appro_rate}
In this section, we derive the average user rate expression of $R_{k,m}$ based on the channel model described in Section II-B, and obtain the corresponding lower bound. This lower bound serves as the basis for subsequent optimization of RIS location.

From (\ref{R_km}), we observe that we only need to calculate $ \| \mathbf{u}_{k,m}\|^{2} $, i.e., $\left[\left( {\mathbf{H}}_{m}^{\text{eff}}\left({\mathbf{H}}_{m}^{\text{eff}}\right)^{H}\right)^{-1}\right]_{k,k}$. Since the expression of $ \| \mathbf{u}_{k,m}\|^{2} $ involves matrix inversion, it is hard to obtain the solution. Therefore, in this paper, we replace $ \| \mathbf{u}_{k,m}\|^{2} $ by finding the average rate. According to \cite{muirhead2009aspects}, the matirx of ${\mathbf{H}}_{m}^{\text{eff}}\left({\mathbf{H}}_{m}^{\text{eff}}\right)^{H}$ follows the non-central Wishart distribution $\mathcal{W}_{K}^{m}\left(N_{t}, \boldsymbol{\mu}_{m}, \boldsymbol{\Sigma}_{m}\right)$, where
\begin{align}
  [\boldsymbol{\mu}_{m}]_{i,:} &=\sqrt{\beta_{1,i}\frac{K_{1}}{K_{1}+1}}\bar{\mathbf{d}}_{i}^{H} \\
  &+ \omega_{i}\sqrt{\beta_{0}\beta_{2,i}\frac{K_{0}K_{2}}{\left(K_{0}+1 \right)\left(K_{2}+1 \right)}}\left( \bar{\mathbf{G}}_{m}\mathbf{\Phi}  \bar{\mathbf{h}}_{i,m} \right)^{H},  \notag  \\  
  [\boldsymbol{\Sigma}_{m}]_{i, j}&=\mathbb{E} \left[ \left({\mathbf{h}}_{i,m}^{\text{eff}} \right)^{H}{\mathbf{h}}_{j,m}^{\text{eff}} \right],  \label{sigma}
\end{align}
wherein $\left(i,j\right) \in \mathcal{K}$. For calculating $ \mathbb{E} \left[ \left({\mathbf{h}}_{i,m}^{\text{eff}} \right)^{H}{\mathbf{h}}_{j,m}^{\text{eff}} \right]$, we have the following Lemma.

\begin{lemma} \label{lemma1}
  The approximate value of $\mathbb{E} \left[ \left({\mathbf{h}}_{i,m}^{\text{eff}} \right)^{H}{\mathbf{h}}_{j,m}^{\text{eff}} \right]$ can be expressed as
  \begin{equation}
  \begin{aligned} \label{res_hh}
    &\mathbb{E} \left[ \left({\mathbf{h}}_{i,m}^{\text{eff}} \right)^{H}{\mathbf{h}}_{j,m}^{\text{eff}} \right] =   \\
    &\left\{\begin{array}{l} 
      \beta_{1,i}+\frac{\omega_{i}\left(K_{0} + K_{2}+1\right)\beta_{0}\beta_{2,i} }{\left(K_{0}+1 \right)\left(K_{2}+1 \right)} +\frac{\omega_{i}K_{0}K_{2}\beta_{0}\beta_{2,i} }{\left(K_{0}+1 \right)\left(K_{2}+1 \right)} \Gamma_{i,j,m}, i = j,\\
      \frac{\omega_{i}\omega_{j}K_{0}K_{2}\beta_{0}\sqrt{\beta_{2,i}\beta_{2,j}} }{\left(K_{0}+1 \right)\left(K_{2}+1 \right)}\Gamma_{i,j,m}, \quad \quad \quad \quad \quad \quad \quad \quad  \text{   } i \neq j,
      \end{array}\right.
   \end{aligned} 
  \end{equation}
   where $ \Gamma_{i,j,m}$ can be expressed as 
   \begin{equation}
   \begin{aligned}
    \Gamma_{i,j,m} = \mathbb{E} \left[ \mathbf{a}\left( \theta_{2,i,m}^{e}, \theta_{2,i,m}^{a}\right) \mathbf{\Phi}^{H}\mathbf{a}^{H}\left( \theta_{0,m}^{e}, \theta_{0,m}^{a}\right) \right.  \\
    \left. \mathbf{a}\left( \theta_{0,m}^{e}, \theta_{0,m}^{a}\right)\mathbf{\Phi}\mathbf{a}^{H}\left( \theta_{2,j,m}^{e}, \theta_{2,j,m}^{a}\right)  \right].
   \end{aligned}
  \end{equation}
\end{lemma}
  
\begin{proof}
Please see Appendix A.
\end{proof}

In this paper, we assume that $\boldsymbol{\Sigma}_{m}$ is an approximate diagonal matrix. This assumption is based on the following reasons. Firstly, since $\beta_{1,i} > \beta_{0}\beta_{2,i}$ \cite{wu2021intelligent}, which makes the values of diagonal elements are much larger than the non-diagonal elements. Secondly, when the differences among $\{ \{  \tilde{\theta}_{2,k}^{a} \}_{k \in \mathcal{K}} ,\tilde{\theta}_{0}^{a} \}$ and $\{ \{  \tilde{\theta}_{2,k}^{e} \}_{k \in \mathcal{K}} ,\tilde{\theta}_{0}^{e}  \}$ are lager, the value of $\Gamma_{i,j,m}$ is approximately equal to zero, which also makes $\boldsymbol{\Sigma}_{m}$ diagonal.

According to \cite{Steyn_Approximations_1972} and \cite[Th. 1]{Kong_Performance_2015}, the non-central Wishart distribution can be transformed into central Wishart distribution. Specifically, $\mathcal{W}_{K}^{m}\left(N_{t}, \boldsymbol{\mu}_{m}, \boldsymbol{\Sigma}_{m}\right)$ is approximated by $\mathcal{W}_{K}^{m}\left(N_{t}, \hat{\boldsymbol{\Sigma}}_{m}\right)$, where $\hat{\boldsymbol{\Sigma}}_{m} = \boldsymbol{\Sigma}_{m} + \frac{1}{N_{t}}\boldsymbol{\mu}_{m}\boldsymbol{\mu}_{m}^{H}$. Then, the average value $\mathbb{E} \left[ \ln \left(\left( {\mathbf{H}}_{m}^{\text{eff}}\left({\mathbf{H}}_{m}^{\text{eff}}\right)^{H}\right)^{-1}\right)_{k,k} \right]$ can be approximated by
\begin{align}
 \ln \left[  \hat{\boldsymbol{\Sigma}}_{m}^{-1} \right]_{k,k} - \psi \left(N_{t} - K + 1\right),  
\end{align}
where $\psi \left( \cdot \right)$ is a digamma function. The average rate of user $k$ at $m$-th subcarrier approximately equals to
\begin{align} \label{bar_rkm}
  \bar{R}_{k,m}=\log _{2}\left(1+\frac{\bar{p}_{k,m}}{\sigma^{2}} \frac{\exp \left(\psi\left(N_{t}-K+1\right)\right)}{\left[\hat{\boldsymbol{\Sigma}}_{m}^{-1}\right]_{k, k}}\right),
\end{align}
where $\bar{p}_{k,m} = \exp \left(\mathbb{E}\left[\ln p_{k,m}\right] \right) $. To calculate $\left[\hat{\boldsymbol{\Sigma}}_{m}^{-1}\right]_{k, k}$, we have the following Lemma.
\begin{lemma} \label{lemma2}
  $\left[\hat{\boldsymbol{\Sigma}}_{m}^{-1}\right]_{k, k}$ can be calculated as 
  \begin{align}
    \left[ \hat{\boldsymbol{\Sigma}}_{m}^{-1} \right]_{k,k} = \frac{1}{\kappa_{k}} - \frac{\tau \frac{|\Xi_{m,k}|^{2}}{\kappa_{k}^{2}}}{1 + \tau \sum_{i=1}^{K}\frac{|\Xi_{m,i}|^{2}}{\kappa_{i}}},
  \end{align}
  where $\kappa_{i}=\beta_{1,i}+\frac{\omega_{i}\left(K_{0} + K_{2}+1\right)\beta_{0}\beta_{2,i} }{\left(K_{0}+1 \right)\left(K_{2}+1 \right)}+\frac{K_{1}\beta_{1,i}}{N_{t}\left(K_{1}+1\right)}$, $\tau = \frac{K_{0}K_{2}}{N_{t}\left(K_{0}+1 \right)\left(K_{2}+1 \right)}$ and $\Xi_{m,i} = \omega_{i}\sqrt{\beta_{0}\beta_{2,i}}\bar{\mathbf{h}}_{i,m}^{H}\mathbf{\Phi}$ $\mathbf{a}^{H}\left( \theta_{0,m}^{e}, \theta_{0,m}^{a}\right)$.
\end{lemma}

\begin{proof}
  $\boldsymbol{\mu}_{m}\boldsymbol{\mu}_{m}^{H}$ can be rewritten as 
  \begin{equation}
  \begin{aligned}
    \boldsymbol{\mu}_{m}\boldsymbol{\mu}_{m}^{H} = &\frac{K_{1}}{K_{1}+1} \text{diag}\left(\beta_{1,1},\cdots,\beta_{1,K}  \right)   \\
    & + \frac{K_{0}K_{2}}{\left(K_{0} + 1 \right)\left(K_{2}+1 \right)} \boldsymbol{\Xi}_{m}\boldsymbol{\Xi}_{m}^{H},
  \end{aligned}
  \end{equation}
  where 
  \begin{equation}
    \boldsymbol{\Xi}_{m}=\left[\begin{array}{c}
      \omega_{1}\sqrt{\beta_{0}\beta_{2,1}}\bar{\mathbf{h}}_{1,m}^{H}\mathbf{\Phi}\mathbf{a}^{H}\left( \theta_{0,m}^{e}, \theta_{0,m}^{a}\right) \\
      \cdots \\
      \omega_{K}\sqrt{\beta_{0}\beta_{2,K}}\bar{\mathbf{h}}_{K,m}^{H}\mathbf{\Phi}\mathbf{a}^{H}\left( \theta_{0,m}^{e}, \theta_{0,m}^{a}\right)
      \end{array}\right] \in \mathbb{C}^{K \times 1}.
  \end{equation}
  Then, the $\hat{\boldsymbol{\Sigma}}_{m}$ can be rewritten as $\tilde{\boldsymbol{\Sigma}}_{m} + \frac{K_{0}K_{2}}{N\left(K_{0}+1 \right)\left(K_{2}+1 \right)} \boldsymbol{\Xi}_{m}\boldsymbol{\Xi}_{m}^{H}$, where $\tilde{\boldsymbol{\Sigma}}_{m} = \boldsymbol{\Sigma}_{m} + \frac{K_{1}}{N(K_{1}+1)} \text{diag}$ $\left(\beta_{1,1},\cdots,\beta_{1,K}  \right)$ is also a diagonal matrix. The $i$-th diagonal element of $ \tilde{\boldsymbol{\Sigma}}_{m} $ is $\beta_{1,i}+\frac{\omega_{i}\left(K_{0} + K_{2}+1\right)\beta_{0}\beta_{2,i} }{\left(K_{0}+1 \right)\left(K_{2}+1 \right)}$ $+\frac{K_{1}\beta_{1,i}}{N_{t}\left(K_{1}+1\right)}$, which is denoted by $\kappa_{i}$. By applying Sherman-Morrison formula, we have 
  \begin{align}
    \hat{\boldsymbol{\Sigma}}_{m}^{-1} = \tilde{\boldsymbol{\Sigma}}_{m}^{-1} - \frac{\tau^{2}\tilde{\boldsymbol{\Sigma}}_{m}^{-1} \boldsymbol{\Xi}_{m}\boldsymbol{\Xi}_{m}^{H}\tilde{\boldsymbol{\Sigma}}_{m}^{-1}}{1 + \tau^{2}\boldsymbol{\Xi}_{m}^{H}\tilde{\boldsymbol{\Sigma}}_{m}^{-1}\boldsymbol{\Xi}_{m}},
  \end{align}
  where $\tau = \frac{K_{0}K_{2}}{N_{t}\left(K_{0}+1 \right)\left(K_{2}+1 \right)}$. The $\left(k,k\right)$-th element of $\hat{\boldsymbol{\Sigma}}_{m}^{-1}$ is obtained as 
  \begin{equation}
  \begin{aligned} \label{sigma_kk}
    \left[\hat{\boldsymbol{\Sigma}}_{m}^{-1}\right]_{k,k}  &= \left[\tilde{\boldsymbol{\Sigma}}_{m}^{-1}\right]_{k,k} - \frac{\tau |\Xi_{m,k}|^{2}\left[\tilde{\boldsymbol{\Sigma}}_{m}^{-1}\right]_{k,k}^{2}}{1 + \tau \sum_{i=1}^{K}|\Xi_{m,i}|^{2}\left[\tilde{\boldsymbol{\Sigma}}_{m}^{-1}\right]_{i,i}}   \\
    &= \frac{1}{\kappa_{k}} - \frac{\tau \frac{|\Xi_{m,k}|^{2}}{\kappa_{k}^{2}}}{1 + \tau \sum_{i=1}^{K}\frac{|\Xi_{m,i}|^{2}}{\kappa_{i}}}.
    \end{aligned}
  \end{equation}
\end{proof}

Substituting (\ref{sigma_kk}) into (\ref{bar_rkm}), we have the average rate $\bar{R}_{k,m}^{\text{AO}}$ of user $k$ at $m$-th subcarrier, which can be expressed as (\ref{appro_rate}).
\begin{figure*}[hb] 
  \hrulefill
  \centering
\begin{equation} \label{appro_rate}
  \bar{R}_{k,m}^{\text{AO}}=\log _{2}\left(1+\frac{\bar{p}_{k,m}\exp \left(\psi\left(N_{t}-K+1\right)\right)\kappa_{k}}{\sigma^{2}}  \left( 1 + \frac{\tau \frac{|\Xi_{m,k}|^{2}}{\kappa_{k}}}{1+ \tau \sum_{i \neq k} \frac{ |\Xi_{m,i}|^{2} }{\kappa_{i}}}  \right)\right).
\end{equation}
\end{figure*}

\begin{corollary}
  The lower bound of the $\bar{R}_{k,m}^{\text{AO}}$ is obtained by 
  \begin{align} \label{lb_rkm}
    \bar{R}_{k,m}^{\text{lb}} = \log _{2}\left(1+\frac{\bar{p}_{k,m}\exp \left(\psi\left(N_{t}-K+1\right)\right)\kappa_{k}}{\sigma^{2}} \right).
  \end{align}
\end{corollary}

\begin{proof}
  Since $ |\Xi_{m,i}|^{2} \geq 0$, the lower bound is obtained.
\end{proof}
In this paper, we adopt an average power allocation strategy, wherein the transmitted power is equally distributed among all users on each subcarrier, i.e., $p_{i,j} = p_{\hat{i},\hat{j}}$ for $ \left( i,\hat{i} \right)\in \mathcal{K}$ and $ \left( j,\hat{j} \right) \in \mathcal{M}$, then different subcarrier has same lower bound rate for the same user, i.e., $\bar{R}_{k,i}^{\text{lb}} = \bar{R}_{k,j}^{\text{lb}}$ for $\left( i,j \right)\in \mathcal{M}$. The lower bound of sum-rate can be expressed as 
\begin{align} \label{lb_rkm_2}
  \bar{R}^{\text{lb}} = M \sum_{k=1}^{K} \tilde{R}_{k}^{\text{lb}},
\end{align}
where $\tilde{R}_{k}^{\text{lb}}$ denotes the lower bound rate at each subcarrier for user $k$.

\begin{corollary}
  Without the RIS, the average rate is 
\begin{equation}
  \begin{aligned} \label{without_ris_rate}
    &\bar{R}_{k,m}^{\text{w/o ris}} =  \\
    &\log _{2}\left(1+\frac{\bar{p}_{k,m}  \psi\left(N_{t}-K+1\right) }{\sigma^{2}}  \frac{\left(N_{t}(K_{1}+1)\beta_{1,k} + K_{1}\beta_{1,k} \right)}{N_{t}(K_{1}+1)}  \right).
  \end{aligned}
\end{equation} 
\end{corollary}

\begin{proof}
  Substituting $\omega_{k} = 0$ into (\ref{appro_rate}), we have (\ref{without_ris_rate}).
\end{proof}

\section{RIS Deployment Optimization}
In this section, we firstly propose a heuristic method by exploiting the lower bound of the average user rate $\bar{R}_{k,m}^{\text{AO}}$, and then analyze the computational complexity and convergence of the proposed algorithm. Finally, the optimization of RIS' phase shifters is also presented.

\subsection{The Proposed Heuristic Method} \label{opt-heu}

In this subsection, we propose a heuristic method to effectively find the optimal location of the RIS.
Based on the lower bound rate (\ref{lb_rkm}) and (\ref{lb_rkm_2}), the problem becomes
\addtocounter{equation}{0}
\begin{IEEEeqnarray}{lcll} 
  \mathcal{P} 3:  &\max _{ \mathbf{p},\varphi_{R}} \quad & \mathbb{E}_{\mathbf{Q}} \left[ \sum_{k=1}^{K} \tilde{R}_{k}^{\text{lb}}\left( \mathbf{p},\varphi_{R},\mathbf{q}_{k}  \right)\right] \label{p3} \\
  &\text{s.t.}  & \mathbf{p} \in \mathcal{S},     \IEEEyessubnumber\label{p3_constraint1}\\
  &             & \varphi_{R} \in [ 0,2\pi).    \IEEEyessubnumber\label{P3_constraint2}
\end{IEEEeqnarray}
We adopt the sample average approximation (SAA) method \cite{Kim_Guide_2015} to deal with the expectation operator, then the objective function becomes $\frac{1}{T}\sum_{t=1}^{T}\sum_{k=1}^{K} \tilde{R}_{k}^{\text{lb}}\left( \mathbf{p},\varphi_{R},\mathbf{q}_{k,t}  \right)$, where $T$ denotes the sample size. Since both index $t$ and $k$ denote different users, by omitting the index $k$, the objective function can further expressed as 
\begin{equation}
\begin{aligned}
  &\frac{1}{T}\sum_{t=1}^{T} \tilde{R}_{t}^{\text{lb}}\left( \mathbf{p}, \varphi_{R},\mathbf{q}_{t}  \right)  \\
   = \quad &\frac{1}{T}\sum_{t=1}^{T} \log _{2}\left(1+\frac{\tilde{p}\exp \left(\psi\left(N_{t}-K+1\right)\right)\kappa_{t}}{\sigma^{2}} \right)   \\
   \le \quad &\frac{1}{T} \log _{2}\left(1+\frac{\tilde{p}\exp \left(\psi\left(N_{t}-K+1\right)\right) \sum_{t=1}^{T} \kappa_{t}}{\sigma^{2}} \right),
\end{aligned} 
\end{equation}
where $\tilde{p} =  \exp \left(\mathbb{E} \left[ \ln \frac{P_{\max}}{KM} \right] \right) $, and $P_{\max}$ is the transmit power at the BS. Then the problem becomes 
\addtocounter{equation}{0}
\begin{IEEEeqnarray}{lcl} 
  \mathcal{P} 3 \mbox{-} 1: \quad \quad &\max _{ \mathbf{p},\varphi_{R}} \quad & \sum_{t=1}^{T} \kappa_{t} \label{p3_1} \\
  &\text{s.t.}  & \mathbf{p} \in \mathcal{S},     \IEEEyessubnumber\label{p3_1_constraint1}\\
  &             &  \varphi_{R} \in [ 0,2\pi ).    \IEEEyessubnumber\label{P3_1_constraint2}
\end{IEEEeqnarray}
In the following, the location of the RIS is found by optimizing $\varphi_R$, $d_0$, $h_0$ and $\varphi_{0}$ respectively.

\textbf{(1) Optimize $\varphi_{R}$}: Due to the expression of $\kappa_{t}$, when the position $\mathbf{p} = [d_0,h_0,\varphi_0]$ is fixed, the problem becomes 
\addtocounter{equation}{0}
\begin{IEEEeqnarray}{lcl} 
  \mathcal{P} 3\mbox{-}2: \quad \quad  &\max _{\varphi_{R}} \quad &\sum_{t=1}^{T}\omega_{t} \beta_{0}\beta_{2,t} \label{p3_2} \\
  &\text{s.t.}  & \varphi_{R} \in [ 0,2\pi ).     \IEEEyessubnumber\label{p3_2_constraint1}
\end{IEEEeqnarray}
Although the orientation of the RIS $\varphi_R$ does not affect the large-scale path loss, it {determines} the number of users served by the RIS, which makes the problem of (\ref{p3_2}) intractable. Here, we heuristically transform the problem into one of maximizing the number of users served by the RIS, which is
\addtocounter{equation}{0}
\begin{IEEEeqnarray}{lcl} 
  \mathcal{P} 3\mbox{-}3: \quad \quad  &\max _{\varphi_{R}} \quad &\sum_{t=1}^{T}\omega_{t} \label{p3_3} \\
  &\text{s.t.}  & \varphi_{R} \in [ 0,2\pi ).     \IEEEyessubnumber\label{p3_3_constraint1}
\end{IEEEeqnarray}
It is means that the design goal of the RIS' orientation $\varphi_R$ is to maximize the number of users served by the RIS. 
For addressing the problem $\mathcal{P} 3\mbox{-}3$, the first step is to define a set of RIS orientation angles, denoted as $\mathcal{N}$, whose size is $N$. Within each RIS orientation angle $\varphi_{R}(i) = 2\pi i/{N} \left(i = 0,1,\dots,N-1\right)$, the SAA method is employed to determine whether a user, based on the geographically generated user location samples, falls within the coverage area of the RIS. The objective function in (\ref{p3_3}) is then computed. Finally, the RIS orientation angle corresponding to the maximum number of served users is selected.

\textbf{(2) Optimize $d_{0}$}: When the other three parameters are fixed, the problem of optimizing $d_{0}$ is 
\addtocounter{equation}{0}
\begin{IEEEeqnarray}{lcl} 
  \mathcal{P} 3\mbox{-}4: \quad \quad  &\max _{d_{0}} \quad &\sum_{t=1}^{T}\omega_{t}\beta_{0}\beta_{2,t} \label{p3_4} \\
  &\text{s.t.}  & d_{0} \in [r_{\min},r_{\max}],     \IEEEyessubnumber\label{p3_4_constraint1}
\end{IEEEeqnarray}
 where $\omega_{t}$ is determined by $\varphi_R$, and $r_{\min}$ and $r_{\max}$ denote the deployment range of the RIS. The large-scale path loss of the BS-RIS-user link  can be expressed as 
\begin{equation}  \label{bru-pathloss}
  \left\{ \begin{array}{ll}
  \beta_{0} &= C_0\left(d_{0}^2 + \left(  h_0 - h_B \right)^2\right)^{-\alpha_{0}/2},  \\
  \beta_{2,t} &= C_0\left((d_{t}^{r})^2 + \left(  h_0 - h_u \right)^2\right)^{-\alpha_{2}/2},  
  \end{array}\right.
\end{equation}
where $C_0$ is the path loss at the reference distance, $\alpha_{0}$ and $\alpha_{2}$ are the path loss exponents for BS-RIS link and RIS-user link, respectively. Although the expression of $d_{t}^{r}$ includes $d_{0}$ based on (\ref{dk_r}), we still regard it as a constant in the process of solving $d_{0}$. Similar as \cite{Pan_Self_2022}, the value of $d_{t}^{r}$ is based on the RIS' location of the previous iteration. Define $\tilde{C}\triangleq C_{0}\sum_{t=1}^{T} w_{t} \beta_{2, t}$, the objective function of (\ref{p3_4}) becomes $y_{1} \triangleq \tilde{C}\left[d_{0}^{2}+\left(h_{0}-h_{B}\right)^{2}\right]^{-\frac{\alpha_{0}}{2}}$. Then, the derivative is 
\begin{align}
  \frac{\partial y_{1}}{\partial d_{0}}=-\alpha_{0} \tilde{C}d_{0}\left[d_{0}^{2}+\left(h_{0}-h_{B}\right)^{2}\right]^{-\frac{\alpha_{0}}{2}-1}. 
\end{align}
It can be observed that $\frac{\partial y_{1}}{\partial d_{0}} <0$. Therefore, we can conclude that 
\begin{align} \label{d0_opt}
  d_{0}^{*} = r_{\min}.
\end{align}
Based on (\ref{d0_opt}), it is surprising {to conclude} that the RIS should be deployed at the closest position to the BS. Moreover, in an end-to-end scenario, the system achieves the highest rate when the RIS is deployed at the position closest to the BS or the user \cite{Wu_Towards_2020,Wu_Intelligent_2019}. In this paper, we consider the distance between RIS and users as a constant, so it is reasonable to get that conclusion of (\ref{d0_opt}).

\textbf{(3) Optimize $h_{0}$}: For optimizing $h_{0}$, we utilize the Dinkelbach’s Transform \cite{Werner_On_1967} to transform the objective function of (\ref{p3_2}) into $\sum_{t=1}^{T} \beta_{0} C_{0} w_{t}-\sum_{t=1}^{T} \left[\left(d_{t}^{r}\right)^{2}+\left(h_{0}-h_{u}\right)^{2}\right]^{\frac{\alpha_{2}}{2}}$, and its lower bound is 
\begin{align}
  \beta_{0} C_{0} \sum_{t=1}^{T} w_{t} -  \left( \sum_{t=1}^{T}\left((d_{t}^{r})^2 + \left(  h_0 - h_u \right)^2\right)\right)^{\alpha_{2}/2}.
\end{align}
By defining $ q_{1} \triangleq C_{0}^{2} \sum_{t=1}^{T} w_{t} > 0 $ and $q_{2} \triangleq \sum_{t=1}^{T}\left(d_{t}^{r}\right)^{2}>0 $, the problem becomes
\addtocounter{equation}{0}
\begin{IEEEeqnarray}{lcl} 
  \mathcal{P} 3\mbox{-}5: \quad \quad  &\max _{ h_{0}} \quad & y_{2} \label{p3_5} \\
  &\text{s.t.}  & h_{0} \in [h_{\min},h_{\max}],     \IEEEyessubnumber\label{p3_5_constraint1}
\end{IEEEeqnarray}
where $y_2 \triangleq q_{1} \left(d_{0}^2 + \left(  h_0 - h_B \right)^2\right)^{-\alpha_{0}/2} -  \left(q_{2} + T  \left(  h_0 - h_u \right)^2\right)^{\alpha_{2}/2}$, $h_{\min}$ and $h_{\max}$ denote the height range of the RIS.
The first-order derivative of $y_{2}$ is 
\begin{equation}
\begin{aligned}
  \frac{\partial y_{2}}{\partial h_{0}}=&-\alpha_{0} q_{1} C_{0}\left(h_{0}-h_{B}\right) \left[d_{0}^{2}+\left(h_{0}-h_{B}\right)^{2}\right]^{-\frac{\alpha_{0}}{2}-1}  \\
   &- \alpha_{2}T\left(h_{0}-h_{u}\right) \left[a_{2}+T\left(h_{0}-h_{u}\right)^{2}\right]^{\frac{\alpha_{2}}{2}-1}.
\end{aligned}
\end{equation}
Based on the fact $h_{B} > h_{u}$, following cases need to be discussed.
\begin{itemize}
  \item When $h_{0} \geq h_{B}$, $\frac{\partial y}{\partial h_{0}} \leq 0$, then the optimal height of the RIS is $h_B$.
  \item When $h_{0}\leq h_{u}$, $\frac{\partial y}{\partial h_{0}}\geq 0$, then the optimal height of the RIS is $h_u$.
  \item When $h_{u}<h_{0}<h_{B}$, we have $\frac{\partial^{2} y_{2}}{\partial d_{0}^2} < 0$ that indicates the objective function increases firstly and then decreases.
\end{itemize}
Based on discussion above, the optimal height is obtained by solving 
\begin{align} \label{h0_opt}
  \frac{\partial y_{2}}{\partial h_{0}}=0.
\end{align}

\textbf{(4) Optimize $\varphi_{0}$}: Similar to optimization of $h_{0}$, based on the Dinkelbach’s Transform \cite{Werner_On_1967}, the objective function is transformed into $\sum_{t=1}^{T} \left[  w_{t} \beta_{0} C_{0} - \left(\left(d_{t}^{r}\right)^{2}+( h_{0} - h_{u} )^{2}\right)^{\alpha_{2} / 2} \right]$. Since the first item of the objective function is a constant, the problem becomes
\addtocounter{equation}{0}
\begin{IEEEeqnarray}{lcl} 
  \mathcal{P} 3\mbox{-}6: \quad \quad  &\min _{ \varphi_{0}} \quad & \sum_{t=1}^{T}\left((d_{t}^{r})^2 + \left(  h_0 - h_u \right)^2\right)^{\alpha_{2}/2} \label{p3_6} \\
  &\text{s.t.}  & \varphi_{0} \in [0,2\pi).     \IEEEyessubnumber\label{p3_6_constraint1}
\end{IEEEeqnarray}
We usually have $(d_{t}^{r})^2 + \left(  h_0 - h_u \right)^2 \gg 1$, then the following inequation is satisfied when $\alpha_{2}/2 > 1$.
\begin{equation}
\begin{aligned}
  & \sum_{t=1}^{T}\left((d_{t}^{r})^2 + \left(  h_0 - h_u \right)^2\right)^{\alpha_{2}/2}  \\
   \le &\left( \sum_{t=1}^{T}\left((d_{t}^{r})^2 + \left(  h_0 - h_u \right)^2\right)\right)^{\alpha_{2}/2}.
\end{aligned}
\end{equation}
The objective function of (\ref{p3_6}) becomes $\min _{ \varphi_{0}} \sum_{t=1}^{T}(d_{t}^{r})^2$, which can be further expressed as $$\min _{ \varphi_{0}} y_{3} \triangleq \sum_{t=1}^{T}-2 d_{0} d_{t}\left[\cos \varphi_{t} \cos \varphi_{0}+\sin \varphi_{t} \sin \varphi_{0}\right].$$
\begin{lemma} \label{lemma3}
  The optimal $\varphi_{0}$ is obtained as 
  \begin{equation} \label{varphi0_opt}
    \varphi_{0}^{*} = 
    \left\{\begin{array}{llll} 
        0,  &y_{3}(0) \leq y_{3}(\pi)  &\text{and} & y_{3}(0) \leq y_{3}(x^{*}),\\
        \pi,  &y_{3}(\pi) \leq y_{3}(0)  &\text{and} & y_{3}(\pi) \leq y_{3}(x^{*}), \\
        x^{*}, &y_{3}(x^{*}) \leq y_{3}(0)  &\text{and} & y_{3}(x^{*}) \leq y_{3}(\pi),
      \end{array}\right.
  \end{equation} 
  where $x^{*}= \arccos \frac{a_{1}}{\sqrt{4 a_{2}^{2}+a_{1}^{2}}}$, and $a_{1} =  \sum_{t=1}^{T}-2 d_{0} d_{t} \cos \varphi_{t}$, $a_{2} = \sum_{t=1}^{T}-2 d_{0} d_{t} \sin \varphi_{t}$.
\end{lemma}
\begin{proof}
  Please see Appendix B.
\end{proof}

\begin{algorithm}[t]
	\caption{The Proposed heuristic Method} 
	\hspace*{0.02in} {\textbf{Input}:} \\
	The distribution of users; the initial location of the RIS $\mathbf{p}_{0}$ \\
	\hspace*{0.02in} {\bf Output:} \\
	The optimized RIS' location $\mathbf{p}^{{h}}$ and the orientation $\varphi^{h}_{R}$
  \begin{algorithmic}[1]
    \While{no convergence}
      \State update $\varphi_R$ by solving (\ref{p3_3}),
      \State update $d_0$ by (\ref{d0_opt}),
      \State update $h_0$ by solving (\ref{h0_opt}),
      \State update $\varphi_0$ by (\ref{varphi0_opt}).
    \EndWhile
    \State \Return $\mathbf{p}^{{h}} = [d_0^{\text{opt}},h_0^{\text{opt}},\varphi_0^{\text{opt}}]$ and $\varphi^{h}_{R} = \varphi_R^{\text{opt}} $  
	\end{algorithmic}
\end{algorithm}
The whole heuristic method is summarized in \textbf{Algorithm 1}. At each iteration, four position parameters are optimized respectively until the sum-rate converges. 

\subsection{Performance Analysis of the Proposed Algorithm}
In this subsection, we summarize the proposed algorithm and analyze its computational complexity and the convergence. Overall, the algorithm proposed in this paper yields suboptimal solutions, primarily due to the decomposition of the original problem and subsequent transformation or scaling of each subproblem.

\subsubsection{Computational Complexity Analysis}
For the proposed heuristic method, the complexity comes from the four subproblems, which involves the complexity in order of $\mathcal{O}\left(NKT\right)$, $\mathcal{O}\left(1\right)$, $\mathcal{O}\left(KT\right)$ and $\mathcal{O}\left(KT\right)$ respectively. Then, the overall complexity of the heuristic method is  $\mathcal{O}\left( N{K}TI_{0}\right)$, where $I_{0}$ is the number of iterations. The computational complexity of each sub-problem is quite low. It is easy to observer the computational complexity of the heuristic scheme does not increase linearly with the number of user samples except for the case of $\varphi_R$ optimization.

As a comparison, we now analyze the computational complexity of using the SGD method. The SGD algorithm selects a user sample at each iteration and updates the parameters based on its gradient. According to \cite{Chen_Reconfigurable_2022}, the form of gradients for $\varphi_{0}$ and $\varphi_{R}$ are complicated and it is hard to make the algorithm converge, hence the optimal $\varphi_{0}$ and $\varphi_{R}$ are obtained by exhaustive search method. At each iteration of the SGD-based algorithm, the complexity is dominated by the generation of $K$ locations and the gradient computation, which have the complexity in order of  $\mathcal{O}\left(K\right)$ and $\mathcal{O}\left(KM\right)$ respectively. Therefore, the overall complexity of the SGD-based algorithm is $\mathcal{O}\left( N^2TKM\right)$, where $N$ denotes the size of the RIS' orientation and azimuth angel set. 

\subsubsection{Convergence Analysis}
Regarding the convergence analysis of the algorithm, we first assume that the number of users served by the RIS lies within the range $[T_{\min}^{\prime},T_{\max}^{\prime}]$. Then, we demonstrate that, when the number of users served by the RIS is fixed, the value of the objective function exhibits both boundedness and monotonicity during the iterative optimization process. Define the number of users served by the RIS is $T^{\prime}$.
For the boundedness, based on the objective function in (\ref{p3_1}), it can be easily found that
\begin{equation}
  \begin{aligned}
    &\sum_{t=1}^{T}\left( \beta_{1,t}+\frac{\omega_{t}\left(K_{0} + K_{2}+1\right)\beta_{0}\beta_{2,t} }{\left(K_{0}+1 \right)\left(K_{2}+1 \right)}+\frac{K_{1}\beta_{1,t}}{N_{t}\left(K_{1}+1\right)} \right) \\
    =&\left( 1 + \frac{K_{1}}{N_{t}\left(K_{1}+1\right)} \right)\sum_{t=1}^{T}\beta_{1,t} + \frac{\left(K_{0} + K_{2}+1\right)\beta_{0} }{\left(K_{0}+1 \right)\left(K_{2}+1 \right)}\sum_{t=1}^{T^{\prime}}\beta_{2,t}.
  \end{aligned}
\end{equation}
Since $\beta_{1,t}$ and $\beta_{2,t}$ represent path loss and their values are usually less than 1, and $\omega_{t} \in \{ 0,1\}$, we have
\begin{equation}
  \left\{\begin{array}{l} 
    \sum_{t=1}^{T} \beta_{1,t} < T \\
    \sum_{t=1}^{T^{\prime}}\beta_{2,t} <T^{\prime}
  \end{array}\right. ,
\end{equation}
It can be easily inferred that the objective function is bounded and less than $\left( 1 + \frac{K_{1}}{N_{t}\left(K_{1}+1\right)} \right)T + \frac{\left(K_{0} + K_{2}+1\right)\beta_{0} }{\left(K_{0}+1 \right)\left(K_{2}+1 \right)}T^{\prime}$.

For the monotonicity, firstly, the distance $d_{0}$ between RIS and BS remains constant due to the optimization result of (\ref{d0_opt}).
Then, in the process of optimization of $h_{0}$ and $\varphi_{0}$, the objective function's lower bound is targeted, which is transformed into $\sum_{t=1}^{T} \beta_{0} C_{0} w_{t}-\sum_{t=1}^{T} \left[\left(d_{t}^{r}\right)^{2}+\left(h_{0}-h_{u}\right)^{2}\right]^{\frac{\alpha_{2}}{2}}$. In the optimization of $\varphi_{0}$, it can be observed that it is equivalent to $\min { \varphi_{0}} \sum_{t=1}^{T}(d_{t}^{r})^2$. By examining the expression of $y_{3}$, it can be inferred that $d_{0}$, $d_{t}$, and $\varphi_{t}$ remain unchanged during the iterative optimization process, so $\sum_{t=1}^{T}(d_{t}^{r})^2$ can be considered a constant during the iteration. Regarding the optimization of $h_{0}$, it can be observed from the expression of $y_2$ that the parameter $q_1$ keeps non-decreasing and $q_{2}$ remains constant during the iteration. Thus, we have $\frac{\partial y_{2}}{\partial q_{1}} \geq 0$, which indicats the non-decreasing property of $y_2$. 

Based above, we can concluded that the objective function value exhibits a non-decreasing property when $T^{\prime}$ is fixed. 
During the iterative process,  $T^{\prime}$ will have a maximum and minimum value, and the corresponding convergence values are denoted as $y_{c}\left(T_{\min}^{\prime}\right)$ and $y_{c}\left(T_{\max}^{\prime}\right)$. As a result, the final convergence value fluctuates within this range. Moreover, simulation results prove that $T^{\prime}$ is fixed after only one or two iterations, which further ensures the convergence of the proposed algorithm.

\subsection{Optimaization for Phase Shifters of the RIS}
Based on the RIS' location optimized in the \ref{opt-heu}, in this subsection, we will describe how the phase shifters of RIS is dynamically adjusted for a specific users' locations. For the optimization of phase shifters of the RIS, there have been many research works, including optimization techniques for continuous phase shifters \cite{huang2019reconfigurable,pan2020multicell,zhang2021robust,huang2020reconfigurable} and for discrete phase shifters \cite{wu2019beamforming,shao2019framework,dai2021reconfigurable}. 

Since the BS adopts ZF precoder, the problem of optimizing $\boldsymbol{\theta}$ is 
\addtocounter{equation}{0}
\begin{IEEEeqnarray}{lcl} 
  \mathcal{P} 4:  & \max _{\boldsymbol{\theta}} \quad &  g\label{p4} \\
  &\text{s.t.}  & | \theta_{n}| = 1, \forall n \in \mathcal{N}_{r}.     \IEEEyessubnumber\label{p4_constraint1}
\end{IEEEeqnarray}
where $g = \sum_{k=1}^{K} \sum_{m=1}^{M} \log _{2}\left(1+\frac{p_{k, m} \left( {\mathbf{h}}_{k,m}^{\text{eff}} \right)^{H} \mathbf{f}_{k, m}  \mathbf{f}_{k, m}^{H} {\mathbf{h}}_{k,m}^{\text{eff}}}{\sigma^{2}}\right)$. According to the quadratic transform \cite{shen_fractional_2018}, the problem becomes
\addtocounter{equation}{0}
\begin{IEEEeqnarray}{lcl} 
  \mathcal{P} 4\mbox{-}1: \quad \quad  & \max _{\boldsymbol{\theta}} & \quad g \left( \boldsymbol{\theta} \right) \label{p4_1} \\
  &\text{s.t.}  & | \theta_{n}| = 1, \forall n \in \mathcal{N}_{r},     \IEEEyessubnumber\label{p4_1_constraint1}
\end{IEEEeqnarray}
where $$ g \left( \boldsymbol{\theta} \right) = \sum_{k=1}^{K} \sum_{m=1}^{M}\left( 2 \Re\left\{\gamma_{k, m}^{H} \left( {\mathbf{h}}_{k,m}^{\text{eff}} \right)^{H} \mathbf{f}_{k, m}\right\}-\gamma_{k, m}^{H} \frac{\sigma^{2}}{p_{k, m}} \gamma_{k, m} \right),$$ with $\gamma_{k, m}$ denoting the introduced auxiliary variable. The whole optimization framework is to alternately update $\gamma_{k, m}$ and $\boldsymbol{\theta}$ until the objective function is convergent. It should be noted that $\mathbf{f}_{k,m}$ is updated based on the ZF principle after obtaining new $\gamma_{k, m}$ and $\boldsymbol{\theta}$.

For given $\boldsymbol{\theta}$, the optimal $\gamma_{k,m}$ can be obtained by solving $\partial g / \partial \gamma_{k, m}=0$, which is 
\begin{align}
  \gamma_{k, m}^{\text{opt}}=\frac{\left( {\mathbf{h}}_{k,m}^{\text{eff}} \right)^{H} \mathbf{f}_{k, m} p_{k, m}}{\sigma^{2}}.
\end{align}

For optimizing $\boldsymbol{\theta}$ based on the given $\gamma_{k,m}$, the problem is 
\addtocounter{equation}{0}
\begin{IEEEeqnarray}{lcl} 
  \mathcal{P} 4\mbox{-}2: \quad \quad  & \max _{\boldsymbol{\theta}} & \sum_{k=1}^{K} \sum_{m=1}^{M}\Re\left\{\gamma_{k, m}^{H} \left( {\mathbf{h}}_{k,m}^{\text{eff}} \right)^{H} \mathbf{f}_{k, m}\right\} \label{p4_2} \\
  &\text{s.t.}  & | \theta_{n}| = 1, \forall n \in \mathcal{N}_{r}.     \IEEEyessubnumber\label{p4_2_constraint1}
\end{IEEEeqnarray}
Define $a_{k,m} \triangleq \mathbf{d}_{k,m}^{H}  \mathbf{f}_{k, m}$ and $\mathbf{v}_{k, m} \triangleq \operatorname{diag}\left(w_{k}  \mathbf{h}_{k,m}^{H}\right) \mathbf{G}_{m}^{H} \mathbf{f}_{k, m}$, we have $\left( {\mathbf{h}}_{k,m}^{\text{eff}} \right)^{H} \mathbf{f}_{k, m} = a_{k, m}+\boldsymbol{\theta}^{H} \mathbf{v}_{k, m}$. Then, the objective function of (\ref{p4_2}) becomes $\sum_{k=1}^{K} \sum_{m=1}^{M} \Re\left\{\boldsymbol{\theta}^{H} \tilde{\mathbf{v}}_{k, m}\right\}$, where $\tilde{\mathbf{v}}_{k, m} = \gamma_{k, m}^{H}\mathbf{v}_{k, m}$. The optimal phase shifters of the RIS is obtained by 
\begin{align} \label{f2}
  \theta_{n}^{\text{opt}}=\frac{\boldsymbol{\nu}_{n}}{\text{abs}\left(\boldsymbol{\nu}_{n}\right)}, \quad \forall n \in \mathcal{N}_{r},
\end{align}
where $\boldsymbol{\nu} = \sum_{k=1}^{K} \sum_{m=1}^{M}\tilde{\mathbf{v}}_{k, m}$, and $\text{abs}\left( \cdot\right)$ returns the absolute value.

In practical scenario, RIS is basically non-ideal due to the limited hardware implementation of metamaterials, which implies phase shifters of the RIS can not be adjusted continuously. Therefore, in this context, we consider the adjustable phase of RIS as discrete, i.e.,
\begin{align}
  \mathcal{F} \triangleq\left\{\theta_{n} \mid \theta_{n} \in\left\{1, e^{j \frac{2 \pi}{L}}, \cdots, e^{j \frac{2 \pi(L-1)}{L}}\right\}\right\},
\end{align}
where $L$ represents the resolution. In this case, the common solution is to find the discrete points closest to the optimal solution \cite{wu2019beamforming}, i.e.,
\begin{align} \label{f3}
  \angle \theta_{n}^{\text {sub }}=\underset{\phi \in \mathcal{F}}{\operatorname{argmin}}\left|\angle \theta_{n}^{\text {opt }}-\angle \phi\right|, \quad  \forall n \in \mathcal{N}_{r},
\end{align}
where $\theta_{n}^{\text {opt }}$ is the optimal solution, can be obtained by (\ref{f2}).

\section{Simulation Results}

In this section, we evaluate the performance of the proposed algorithm by running numerical simulations in the wideband RIS-assisted mmWave MIMO systems. The default system parameters are shown in Table \ref{table1}.

\begin{table}[!hpt]
  \caption{Simulation parameters}
  \label{table1}
  \centering
  \begin{tabular}{lr}  \toprule
    \textbf{Simulation Parameters}      & \textbf{Value}  \\ \midrule
    The number of antennas at BS$N_{t}$ & 128            \\ 
    The number of RIS elements $N_{r}$  & 100           \\
    The number of subcarriers $M$       & 16             \\
    The number of users $K$             & 4             \\
    The frequency of carrier $f_c$      & 28GHz          \\
    The bandwidth $B$                   & 4GHz          \\
    The height of the BS $h_{B}$        & 10m             \\
    The height of users $h_u$           & 1.5m          \\
    The radius of the cell $r$          & 200m           \\ 
    The transmit power at the BS $P_{\max}$   & 30 dBm        \\
    The noise power  $\sigma^2$    & -104dBm         \\ \bottomrule
  \end{tabular}
\end{table}

For the channel paremeters, the path loss of the BS-user link is $\beta_{1, k}=C_{1} d_{k}^{-\alpha_{1}}$, where $C_{1} = \frac{\lambda^{2}}{16 \pi^{2}} $ is based on \cite{tang_path_2022}, $\alpha_{1} = 4$, and $\lambda$ is the wave length. Then the path loss of the BS-RIS link and RIS-user link are shown in (\ref{bru-pathloss}). The Rician factors of the BS-RIS link, BS-user link, and RIS-user link are set to $K_{0} = 15$, $K_{1} = 10$ and $K_{2} = 15$. Moreover, in this paper, we assume that the BS is located in the center of the cell, and consider three geographical distributions of users to illustrate the advantage of the propsoed scheme, which are 
\begin{itemize}
  \item \textbf{Uniform Distribution}: Users are uniform distributed in a circular cell with radius $r$.
  \item \textbf{One Hotspot}: The cell {contains} one hotspot region, and users are randomly distributed in a circle centered at polar coordinate of $(50, \pi/4)$ with radius $10$ m.
  \item \textbf{Multiple Hotspots}: The cell {contains} four hotspot regions whose center polar coordinates are $(50, \frac{\pi}{4})$, $(100, \frac{3\pi}{4})$, $(50, -\frac{3\pi}{4})$, $(100, -\frac{\pi}{4})$, respectively. Users are randomly distributed in these centers with radius $10$ m.
\end{itemize}

Fig. \ref{fig-approximation_vs_power} illustrates the sum-rate against the transmit power at the BS, and compares the sum-rate of simulation, approximation and lower bound results. In each channel realization, the location of the RIS is randomly selected within the cell range. The figure indicates that the average user rate is quite close to the simulation rate in all three distributions, verifying the derivation in Section-III. However, the level of tightness is different. In \textit{one hotspot} distribution, the gap between simulation and approxiamtion results is smallest. This can be explained by the fact that $\boldsymbol{\Sigma}_{m}$ is more like a diagonal matrix in \textit{one hotspot} distribution {than in the other} two distributions.
\begin{figure}[hbtp]
  \centering
  \includegraphics[width = 80mm]{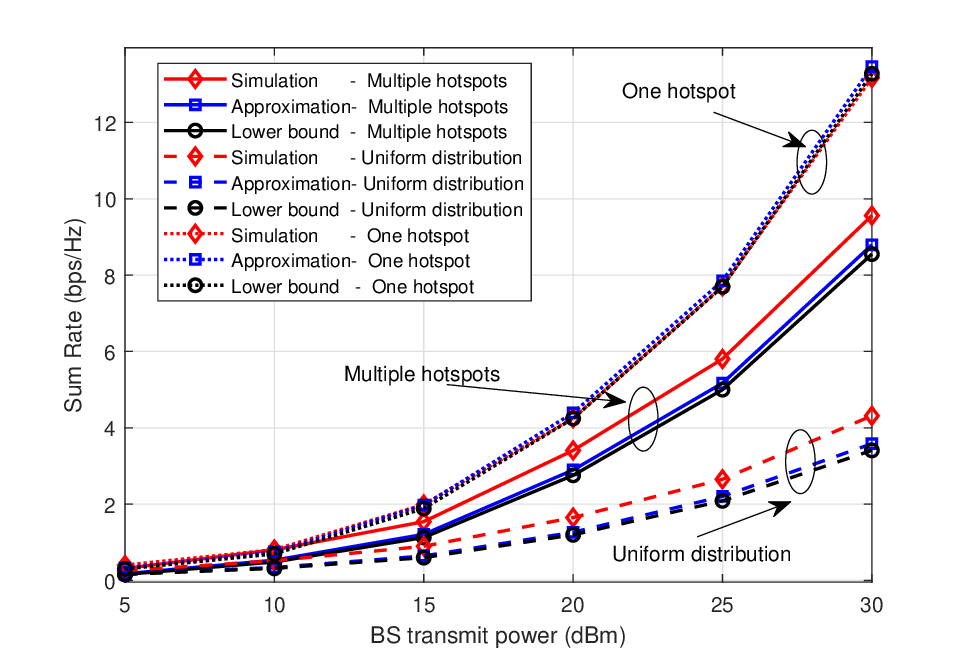}
  \caption{Sum-rate versus the transmit power at BS.}
  \label{fig-approximation_vs_power}
\end{figure}

We use the following deployment strategies for performance comparison, which are
\begin{itemize}
  \item \textbf{Exhaustive search}: Exhaustive search for the optimal RIS position by varying its four position parameters.
  \item \textbf{SGD-based algorithm}: Based on \cite{Chen_Reconfigurable_2022}, using the SGD method to obtain the RIS' location parameters. Specifically, in each iteration, a user location sample is generated, and the gradient of the RIS' location parameters is used for updating.
  \item \textbf{Heuristic}: The proposed algorithm in this paper.
  \item \textbf{Heuristic - one sample}: In order to demonstrate the necessity of considering the long-term geographic distribution of users, we added a comparison curve using a heuristic algorithm but with only one sample of uers locations.
  \item \textbf{Random}: The position parameters of the RIS are randomly selected.
\end{itemize}

Now, we want to evaluate the performance of the proposed RIS deployment strategies. Firstly, the optimized results of the RIS' location by exhaustive search, SGD-based and heuristic algorithms in three user distributions are presented in Table \ref{table2}. Then, Fig. \ref{fig-deployment_power}  compares the performance of different deployment strategies under different user geographical distributions. 
When the transmitted power is $25$ dBm, the sum-rate of the heuristic method is $7.7$ bps/Hz, $13.77$ bps/Hz, and $5.06$ bps/Hz higher than the random deployment strategy in \textit{multiple hotspots}, \textit{one hotspot}, and \textit{uniform distribution}, respectively. 
In percentage terms, the performance of proposed algorithm is 142.01\%, 182.65\%, and 215.81\% higher than that of random deployment strategy, and 10.89\%, 4.23\%, and 6.63\% lower than the performance of the exhaustive search approach.
From the figures, it can be seen that the performance of the proposed heuristic method is close to the exhaustive search scheme in all three distributions, and it is better than the SGD-based algorithm. Besides, the performance gap between the heuristic method and the SGD-based algorithm is clear in \textit{one hotspot} distribution, while the gap becomes quite small in \textit{uniform distribution}. The performance gap between heuristic and heuristic-one sample method is obvious in the \textit{one hotspot} distribution. This is because in the \textit{one hotspot} distribution, the RIS with its location optimizing by the heuristic method can serve all users in the area, however, the method with using only one user sample does not consider the situation of users in other locations, which may result in the RIS not covering some users in certain locations and causing significant performance loss. Moreover, the system achieves highest sum-rate in \textit{one hotspot} distribution, which indicates the RIS is better to serve users in \textit{one hotspot}.

\begin{table*}[!hpt]
  \caption{Optimized results of the RIS' location}
  \label{table2}
  \centering
  \resizebox{\textwidth}{10mm}{
  \begin{tabular}{llll}  \toprule
    \textbf{Locations}      & \textbf{The exhaustive search}  &   \textbf{The SGD-based algorithm}            & \textbf{The heuristic algorithm}\\ \midrule
    \textbf{Multiple hotshots distribution} & $\mathbf{p}=\left[10,2.50,7\right],\varphi_R=0.76$ & $\mathbf{p}=\left[10,1,1\right],\varphi_R = -0.74$ & $\mathbf{p} =\left[10,1.13,10\right],\varphi_R=0.80$ \\
    \textbf{One hotshot distribution}      & $\mathbf{p} =\left[10,1.00,5\right], \varphi_R=0.96$ & $\mathbf{p}=\left[10,1.2,1\right],\varphi_R = 1.06$ & $\mathbf{p} =\left[10,1.10,10\right],\varphi_R = 0.89$ \\
    \textbf{Uniform distribution}      & $\mathbf{p}=\left[10,3.10,5\right],\varphi_R = 2.66$ & $\mathbf{p} =\left[10,0.5,1\right],\varphi_R = -0.14$ & $\mathbf{p} =\left[10,2.03,10\right],\varphi_R = 0.79$ \\
   \bottomrule
  \end{tabular}
  }
\end{table*}

\begin{figure}
  \centering
  \subfigure[Multiple hotspots]{\label{fig-deployment_power_MH}
  \includegraphics[width=80mm]{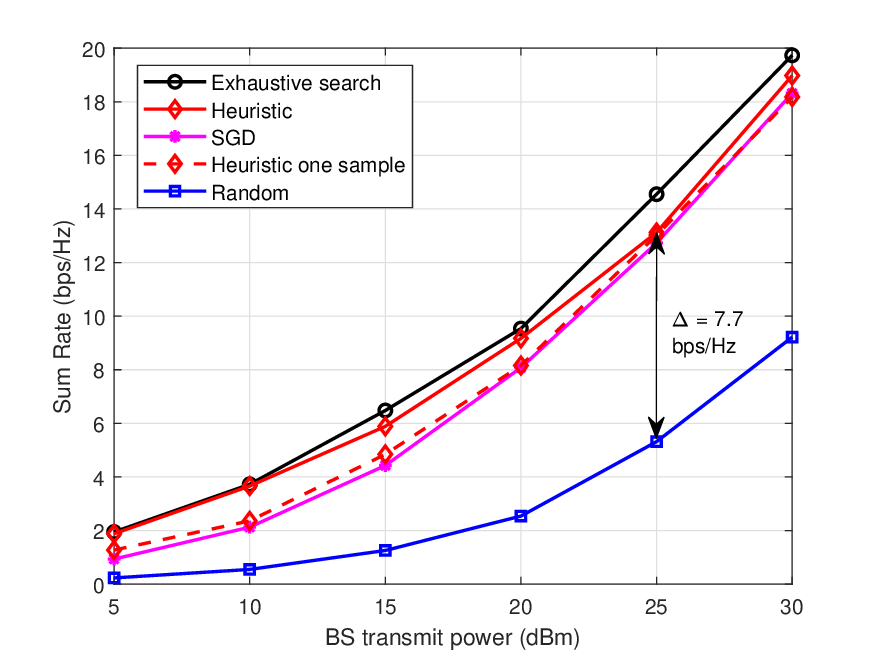}}
  \vspace{0.01\linewidth}
  \subfigure[One hotspot]{\label{fig-deployment_power_OH}
  \includegraphics[width=80mm]{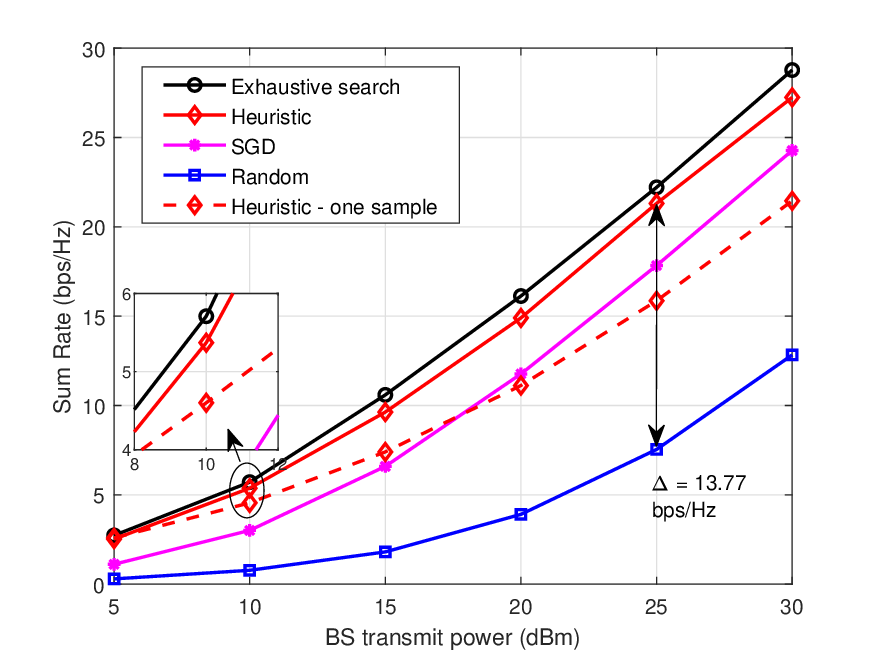}}
  \vspace{0.01\linewidth}
  \subfigure[Uniform distribution]{\label{fig-deployment_power_UD}
  \includegraphics[width=80mm]{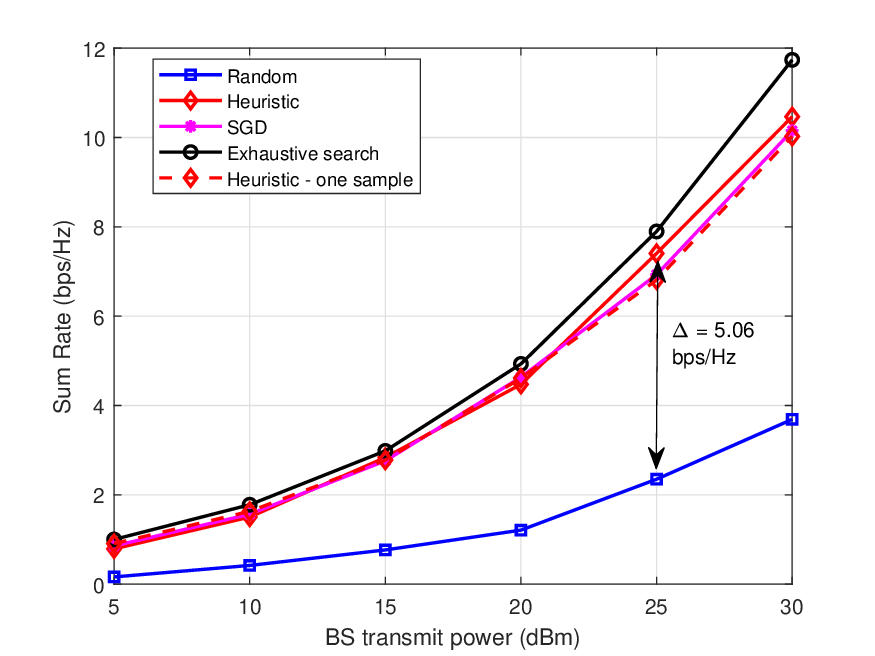}}
  \caption{Different deployment strategies comparison, sum-rate versus transmit power.}
  \label{fig-deployment_power}
\end{figure}

Fig. \ref{fig-deployment_power_Nr} illustrates the sum-rate against the number of RIS' phase shifters for different deployment strategies in all three distributions. 
Compared with the random deployment strategy, the sum-rate gain of the heuristic method is $4.17$ bps/Hz,  $6.82$ bps/Hz, and  $2.01$ bps/Hz in \textit{multiple hotspots}, \textit{one hotspot}, and \textit{uniform distribution}, respectively, when $N_r = 196$. The performance loss is $0.43$ bps/Hz, $0.77$ bps/Hz, and  $0.31$ bps/Hz when compared with the exhaustive search scheme. For the random deployment of the RIS, the performance improvement is limited with the increase of the number of RIS' phase shifters. Moreover, the performance gap between exhaustive search scheme and the heuristic method or between the heuristic method and random deployment strategy becomes larger as the number of the RIS' phase shifters increase. However, in the \textit{uniform distribution}, the performance of the heuristic, SGD-based and heuristic-one sample methods is similar. This is because in the \textit{uniform distribution}, the number of users that RIS can cover in each direction is roughly the same, and for the decisive parameter $d_{0}$, these three methods can also reach the same conclusion, which is, $d_{0} = r_{\min}$.
\begin{figure}
  \centering
  \subfigure[Multiple hotspots]{\label{fig-deployment_power_Nr_MH}
  \includegraphics[width=80mm]{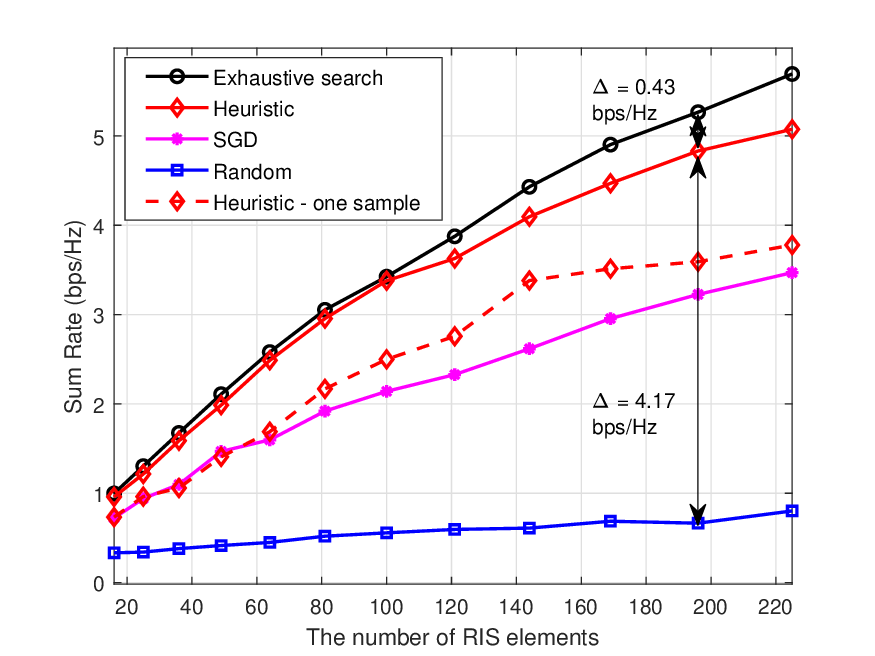}}
  \vspace{0.01\linewidth}
  \subfigure[One hotspot]{\label{fig-deployment_power_Nr_OH}
  \includegraphics[width=80mm]{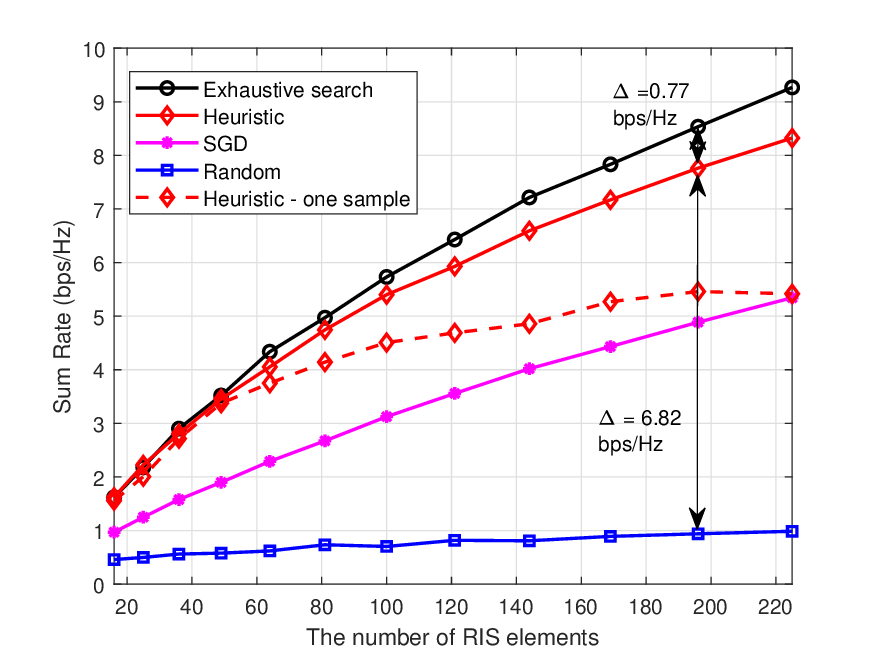}}
  \vspace{0.01\linewidth}
  \subfigure[Uniform distribution]{\label{fig-deployment_power_Nr_UD}
  \includegraphics[width=80mm]{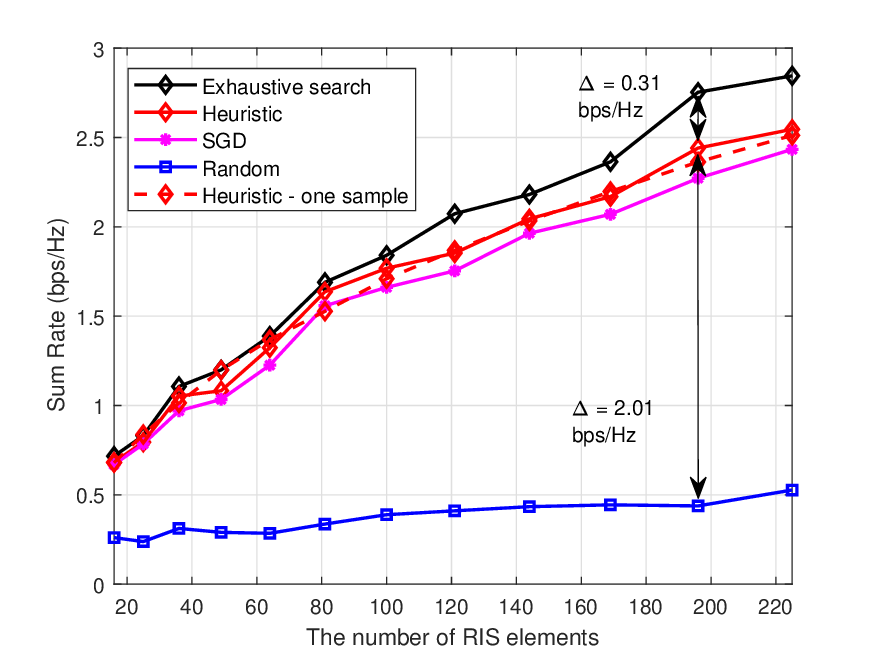}}
  \caption{Different deployment strategies comparison, sum-rate versus $N_{r}$.}
  \label{fig-deployment_power_Nr}
\end{figure}

Fig. \ref{fig-other-para} illustrates the impact of other system parameters on performance. 
In Fig. \ref{fig-Nt}, it can be seen that as the number of antennas increases, the sum-rate exhibit a growth trend, which is consistent with the variation pattern in traditional MIMO wireless systems. In Fig. \ref{fig-users}, in order to further demonstrate the impact of the number of users on performance, we set the number of antennas at the BS to 64, and other parameters remain the same as in Fig. \ref{fig-deployment_power}. Since the transmitted power at the BS is limited and the power allocation strategy is fixed, it can be observed that in the \textit{multiple hotspots} distribution, the performance of all methods shows a trend of first increasing and then decreasing, which is consistent with the trend in traditional wireless systems. In scenarios with a small number of users, the BS can provide adequate streams, and it equips a sufficient number of antennas to eliminate multi-user interference. Consequently, system performance increases with an increasing number of users. However, as the number of users  gradually grows, the BS may not completely eliminate interference, and inter-user interference becomes the primary factor limiting system performance. At this point, system performance decreases with an increasing number of users. In addition, the best point of users' number of the heuristic curve is larger than that of SGD, indicating that the correct deployment of RIS can accommodate more users and achieve higher sum-rate. In Fig. \ref{fig-mmse}, the performance comparison of the BS using ZF and minimum mean square error (MMSE) precoder is shown for the \textit{one hotspot} distribution. The RIS deployment optimization for MMSE precoder is provided in Appendix \ref{mmse-solution}. From the picture, it can be observed that the performance of MMSE precoder is slightly superior to that of ZF precoder for all algorithms. This is attributed to MMSE precoder's consideration of noise impact, leading to performance improvement due to the precoder design rather than performance gains from optimizing the location of RIS.
\begin{figure}
  \centering
  \subfigure[Sum rate versus the number of antennas at the BS]{\label{fig-Nt}
  \includegraphics[width=80mm]{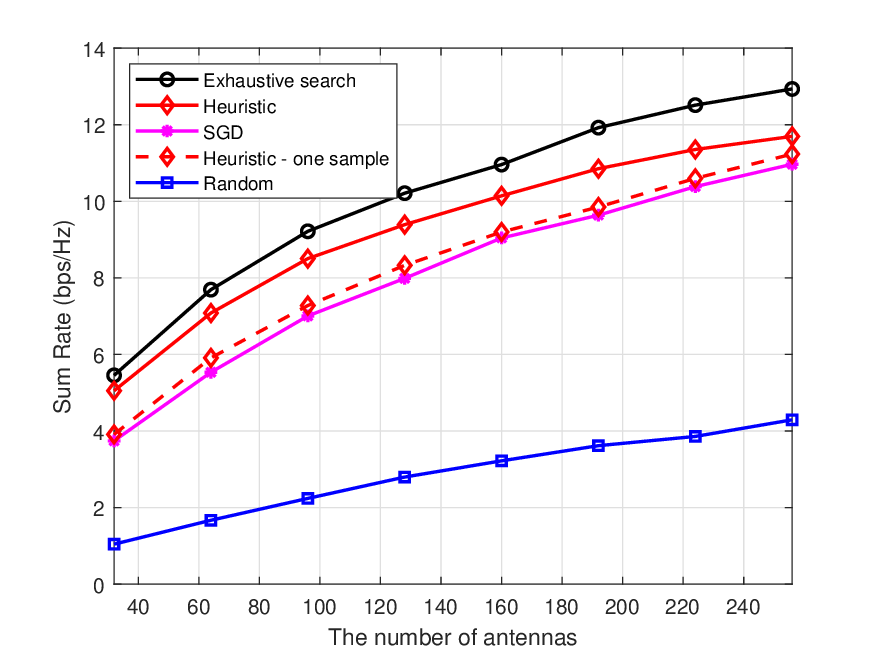}}
  \vspace{0.01\linewidth}
  \subfigure[Sum-rate versus the number of users]{\label{fig-users}
  \includegraphics[width=80mm]{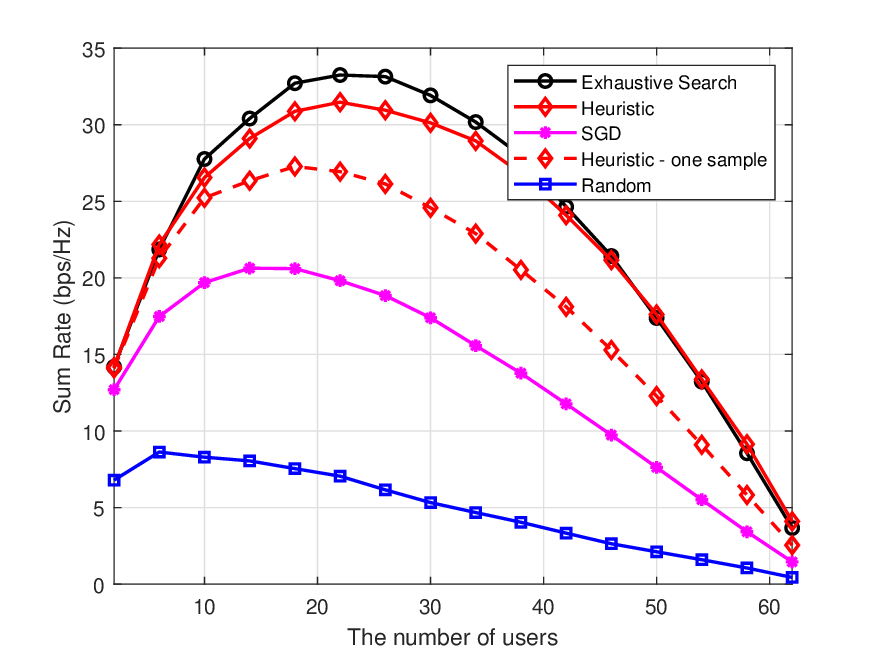}}
  \vspace{0.01\linewidth}
  \subfigure[ZF and MMSE precoder comparison]{\label{fig-mmse}
  \includegraphics[width=80mm]{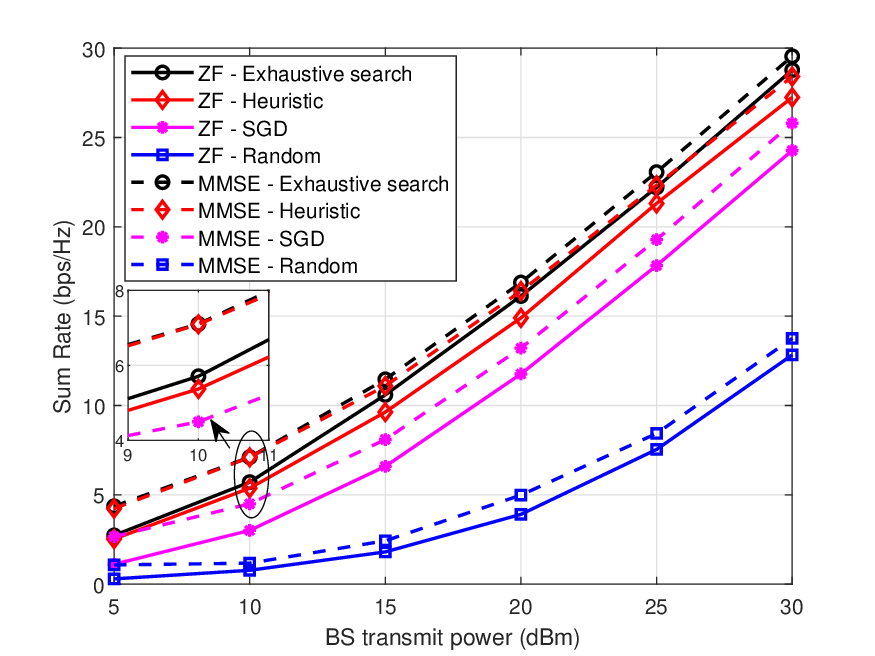}}
  \caption{Impact of other system parameters on performance.}
  \label{fig-other-para}
\end{figure}

Fig. \ref{fig-iterations} demonstrates the sum-rate as the number of user location samples varies. The curve for the SGD algorithm reflects its convergence, while the curve for the heuristic algorithm reflects the necessity of considering the users' long-term geographical distribution. 
From the figure, we know that the SGD-based algorithm can converge within at least $80$ user location samples, and in order to facilitate faster convergence, we set step size of the SGD-based algorithm to a larger value, such as $1$. However, as reported in \cite{Diederik_Adam_2015}, the step size is typically set to $10^{-4}$, indicating that a large number of user samples are needed.

For the heuristic algorithm, the sum-rate increases sharply with the increase in the number of samples when the number of samples is small.
 
However, as the number of samples continues to increase, the sum-rate tends to stabilize. Additionally, it can also be observed that the number of user location samples required for the heuristic algorithm is smaller than that required for convergence of the SGD algorithm, which indirectly reflects the low computational complexity of the heuristic algorithm. Moreover, we also add the convergence curve of the heuristic algorithm, where the x-axis represents the number of iterations. The heuristic method can converge after only $3$ iterations, while the SGD-based algorithm requires at least $80$ user samples, which highlights the low complexity of the heuristic method from another aspect.

\begin{figure}[hbtp]
  \centering
  \includegraphics[width=80mm]{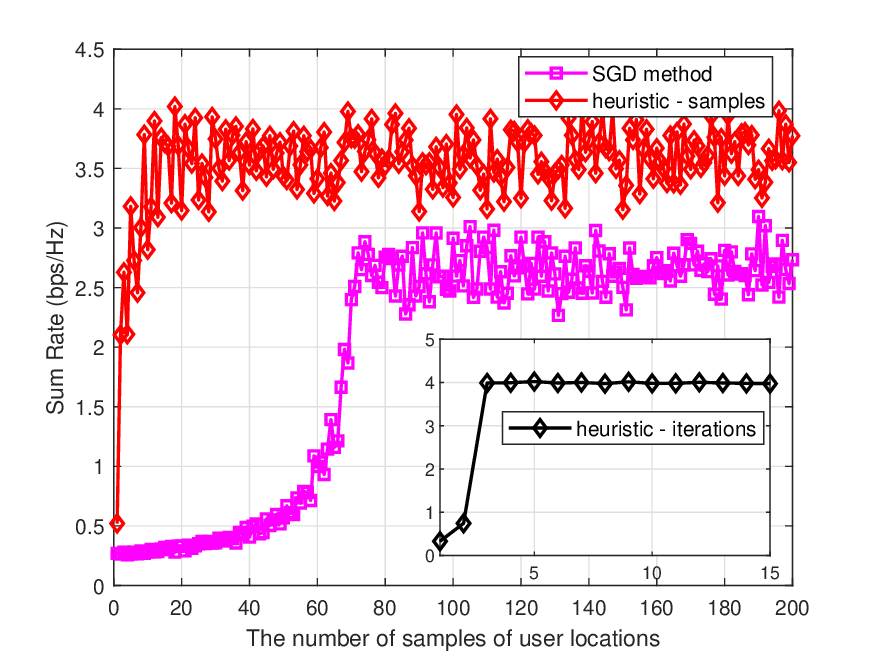}
  \caption{Sum-rate versus the number of samples of user locations and iterations.}
  \label{fig-iterations}
\end{figure}

Fig. \ref{fig-variation} depicts the impact of several location parameters on sum-rate, where Fig. \ref{fig-variation_do} describes the sum-rate against the distance between the BS and the RIS $d_0$, and Fig. \ref{fig-variation_phiR}  illustrates the impact of RIS' orientation on sum-rate. It should be noted that remaining parameters are kept constant according to the optimization results of the heuristic method. From \ref{fig-variation_do}, in all three distributions, the sum-rate exhibits a decreasing trend with increasing $d_0$, which is consistent with the previous analysis. Moreover,  the curves of \textit{multiple hotspots} and \textit{one hotspot} distributions have a temporary peak around at $d_0 = 50$, which is because the RIS is close to the hotspot area in the first quadrant when $d_0$ is around $50$. At this moment, the RIS is close to users, resulting in a temporary increase in the sum-rate. From this figure, it is more obvious that the sum-rate in the \textit{one hotspot} distribution is higher than the other two distributions.

Fig. \ref{fig-variation_phiR} not only draws the curves of sum-rate against the orientation of RIS, but also illustrates the RIS orientation optimization results based on the heuristic algorithms in \textit{multiple hotspots, one hotspot, and uniform} distribution, respectively, using diamond, pentagram, and square markers. We can observe that the optimized orientation is a sub-optimal solution since we do not directly solve the problem $\mathcal{P} 3\mbox{-}2$ when optimizing the orientation of RIS. Instead, we adopt heuristics to maximize the number of users served by the RIS. Besides, the sum-rate is very low at some orientations, because the RIS cannot receive the signal from the BS at these orientations, i.e., this is the system performance without the help of the RIS. In \textit{multiple hotspots} and \textit{one hotspot} distributions, the selection of the RIS' orientation affects the sum-rate to a some extent, while in the \textit{uniform distribution}, the curve is almost a straight line. This is because the number of users that the RIS can serve in all directions is almost the same in the \textit{uniform distribution}, resulting in the sum-rate remaining almost constant.
\begin{figure}
  \centering
  \subfigure[Sum-rate versus $d_{0}$]{\label{fig-variation_do}
  \includegraphics[width=80mm]{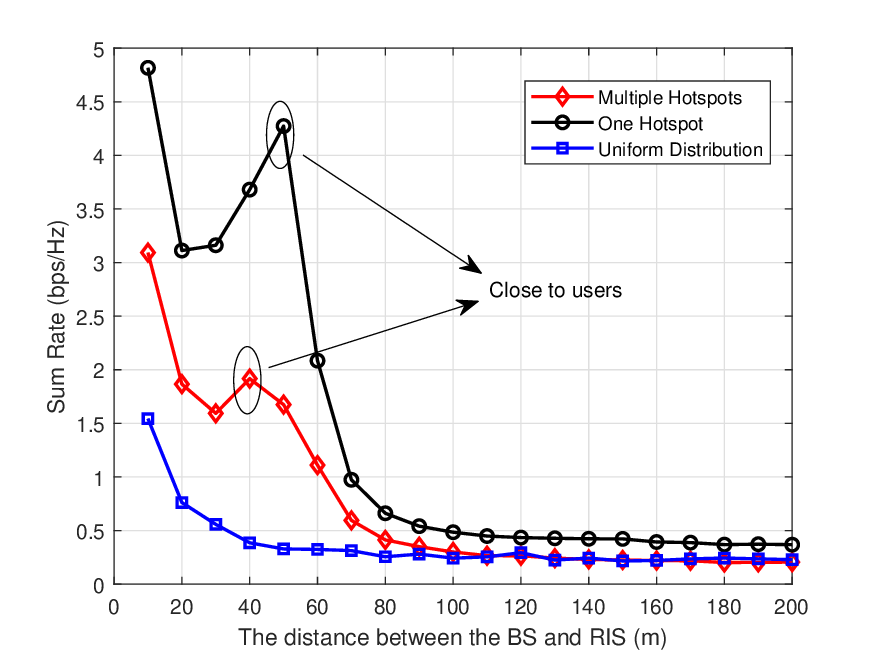}}
  \vspace{0.01\linewidth}
  \subfigure[Sum-rate versus the rotation of the RIS]{\label{fig-variation_phiR}
  \includegraphics[width=80mm]{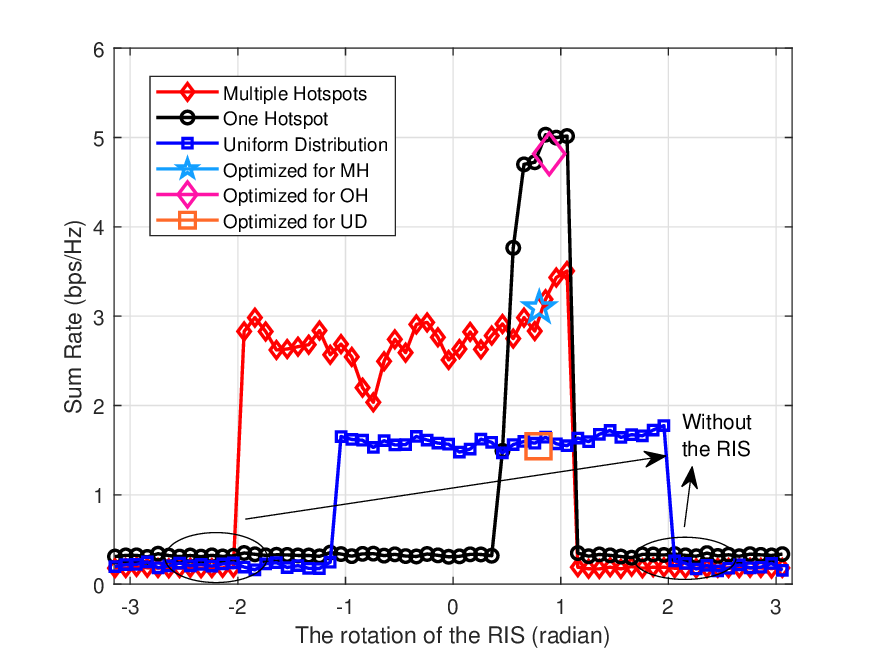}}
  \caption{Variation of sum-rate with parameters of the RIS' location.}
  \label{fig-variation}
\end{figure}

Fig. \ref{fig-theta_opt} verifies the performance improvement that can be brought by optimizing phase shifters of the RIS. 
In Fig. \ref{fig-theta_opt}, performance curves of continuous phase shifters, and discrete phase shifters with 1-bit resolution and 3-bit resolution are presented. From the figure, it can be observed that the performance improvement of continuous phase shifters over random phase shifters is $2.3$ bps/Hz at 20 dBm. The performance of discrete phase shifters with different resolutions is in between that of continuous phase shifters and random phase shifters, where the performance of 3-bit phase shifters is close enough to that of continuous phase shifters, providing a useful reference for practical selection of the RIS.
\begin{figure}[hbtp]
  \centering
  \includegraphics[width=80mm]{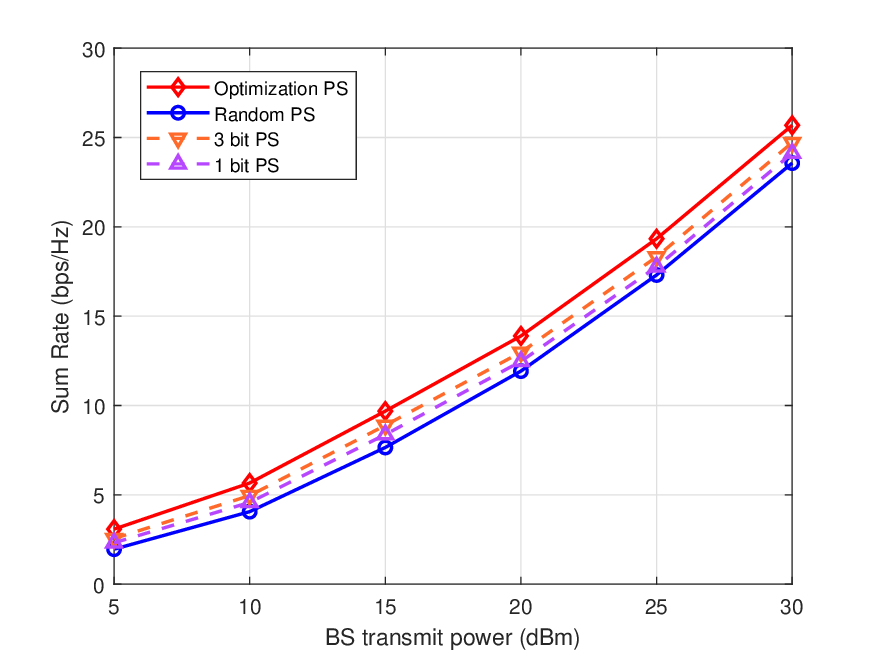}
  \caption{The optimization of phase shifters of the RIS.}
  \label{fig-theta_opt}
\end{figure}

\section{Conclusion}
In this paper, we propose a RIS deployment strategy for the wideband RIS-assisted mmWave MIMO system, based on the long-term geographic distribution of users, with the goal of maximizing the sum-rate of the cell. Firstly, we derive the average user rate of $R_{k,m}$, as well as the lower bound rate, by exploiting the characteristic that the covariance of the channel follows the Wishart distribution. Then, we propose a heuristic method based on the lower bound rate, where the RIS' orientation optimization problem is transformed into the problem of maximizing the number of users served by the RIS. It is concluded that the RIS should be deployed close to the BS. Finally, we analyze the computational complexity and convergence of the proposed algorithm and optimize the phase shifters of the RIS based on the specific locations of users. The simulation results show that the proposed deployment strategy can effectively improve the sum-rate of the cell, and the RIS is more suitable for serving one hotspot area. In our future work, we will study the deployment of multiple RISs in multiple cells with different users' geographic distributions by utilizing ML based methods.

\appendices
\section{Proof of Lemma \ref{lemma1}}
Since the elements of $\tilde{\mathbf{G}}_{m}$, $\tilde{\mathbf{d}}_{k,m}$ and $\tilde{\mathbf{h}}_{k,m}$ are the i.i.d complex Gaussian random variables, we have 
\begin{equation}
  \left\{\begin{array}{l} 
    \mathbb{E}\left[\tilde{\mathbf{G}}_{m}\right]=\mathbf{0}_{N_{t} \times N_{r}},\\
    \mathbb{E}\left[\tilde{\mathbf{d}}_{k,m}\right]=\mathbf{0}_{N_{t} \times 1}, \\
    \mathbb{E}\left[\tilde{\mathbf{h}}_{k,m}\right]=\mathbf{0}_{N_{r} \times 1}.
    \end{array}\right.
\end{equation} 
Then, the (\ref{sigma}) can be calculated as (\ref{hh}), which is shown in the bottom of the next page.
\begin{figure*}[hb] 
  \hrulefill
  \centering
\begin{equation}
\begin{aligned}	\label{hh}     
&\mathbb{E} \left[ \left({\mathbf{h}}_{i,m}^{\text{eff}} \right)^{H}{\mathbf{h}}_{j,m}^{\text{eff}} \right]  \\
=&\frac{K_{1}\sqrt{\beta_{1,i}\beta_{1,j}}}{K_{1}+1}\mathbb{E} \left[ \bar{\mathbf{d}}_{i,m}^{H}\bar{\mathbf{d}}_{j,m} \right] +\omega_{j}\sqrt{\frac{\beta_{0}\beta_{1,i}\beta_{2,j}K_{0}K_{1}K_{2}}{\left(K_{0}+1 \right)\left(K_{1}+1 \right)\left(K_{2}+1 \right)}}\mathbb{E} \left[ \bar{\mathbf{d}}_{i,m}^{H} \bar{\mathbf{G}}_{m}\mathbf{\Phi}  \bar{\mathbf{h}}_{j,m} \right]  \\
+& \frac{\sqrt{\beta_{1,i}\beta_{1,j}}}{K_{1}+1}\mathbb{E} \left[\tilde{\mathbf{d}}_{i,m}^{H}\tilde{\mathbf{d}}_{j,m} \right] +\omega_{i}\sqrt{\frac{\beta_{0}\beta_{1,i}\beta_{2,j}K_{0}K_{1}K_{2}}{\left(K_{0}+1 \right)\left(K_{1}+1 \right)\left(K_{2}+1 \right)}}\mathbb{E} \left[\bar{\mathbf{h}}_{i,m}^{H} \mathbf{\Phi}^{H} \bar{\mathbf{G}}_{m}^{H} \bar{\mathbf{d}}_{j,m}\right]  \\ 
+&\frac{\omega_{i}\omega_{j}K_{0}K_{2}\beta_{0}\sqrt{\beta_{2,i}\beta_{2,j}} }{\left(K_{0}+1 \right)\left(K_{2}+1 \right)}\mathbb{E} \left[\bar{\mathbf{h}}_{i,m}^{H} \mathbf{\Phi}^{H} \bar{\mathbf{G}}_{m}^{H}\bar{\mathbf{G}}_{m}\mathbf{\Phi}  \bar{\mathbf{h}}_{j,m} \right] + \frac{\omega_{i}\omega_{j}K_{0}\beta_{0}\sqrt{\beta_{2,i}\beta_{2,j}} }{\left(K_{0}+1 \right)\left(K_{2}+1 \right)}\mathbb{E} \left[\tilde{\mathbf{h}}_{i,m}^{H} \mathbf{\Phi}^{H} \bar{\mathbf{G}}_{m}^{H}\bar{\mathbf{G}}_{m}\mathbf{\Phi}  \tilde{\mathbf{h}}_{j,m} \right]  \\
+&\frac{\omega_{i}\omega_{j}K_{2}\beta_{0}\sqrt{\beta_{2,i}\beta_{2,j}} }{\left(K_{0}+1 \right)\left(K_{2}+1 \right)}\mathbb{E} \left[\bar{\mathbf{h}}_{i,m}^{H} \mathbf{\Phi}^{H} \tilde{\mathbf{G}}_{m}^{H}\tilde{\mathbf{G}}_{m}\mathbf{\Phi}  \bar{\mathbf{h}}_{j,m} \right] +\frac{\omega_{i}\omega_{j}\beta_{0}\sqrt{\beta_{2,i}\beta_{2,j}} }{\left(K_{0}+1 \right)\left(K_{2}+1 \right)}\mathbb{E} \left[\tilde{\mathbf{h}}_{i,m}^{H} \mathbf{\Phi}^{H} \tilde{\mathbf{G}}_{m}^{H}\tilde{\mathbf{G}}_{m}\mathbf{\Phi}  \tilde{\mathbf{h}}_{j,m} \right].
  \end{aligned}
\end{equation}
\end{figure*}

Now we should calculate $\mathbb{E} \left[ \bar{\mathbf{d}}_{i,m}^{H}\bar{\mathbf{d}}_{j,m} \right]$, $\mathbb{E} \left[ \tilde{\mathbf{d}}_{i,m}^{H}\tilde{\mathbf{d}}_{j,m} \right]$, $\mathbb{E} \left[ \bar{\mathbf{d}}_{i,m}^{H} \bar{\mathbf{G}}_{m}\mathbf{\Phi}  \bar{\mathbf{h}}_{j,m} \right]$, $\mathbb{E} \left[\bar{\mathbf{h}}_{i,m}^{H} \mathbf{\Phi}^{H} \bar{\mathbf{G}}_{m}^{H} \bar{\mathbf{d}}_{j,m}\right]$, $\mathbb{E} \left[\bar{\mathbf{h}}_{i,m}^{H} \mathbf{\Phi}^{H} \bar{\mathbf{G}}_{m}^{H}\bar{\mathbf{G}}_{m}\mathbf{\Phi}  \bar{\mathbf{h}}_{j,m} \right]  $, $ \mathbb{E} \left[\tilde{\mathbf{h}}_{i,m}^{H} \mathbf{\Phi}^{H} \bar{\mathbf{G}}_{m}^{H}\bar{\mathbf{G}}_{m}\mathbf{\Phi}  \tilde{\mathbf{h}}_{j,m} \right] $, $ \mathbb{E} \left[\bar{\mathbf{h}}_{i,m}^{H} \mathbf{\Phi}^{H} \tilde{\mathbf{G}}_{m}^{H}\tilde{\mathbf{G}}_{m}\mathbf{\Phi}  \bar{\mathbf{h}}_{j,m} \right] $ and $\mathbb{E} \left[\tilde{\mathbf{h}}_{i,m}^{H} \mathbf{\Phi}^{H} \tilde{\mathbf{G}}_{m}^{H}\tilde{\mathbf{G}}_{m}\mathbf{\Phi}  \tilde{\mathbf{h}}_{j,m} \right]$, respectively.

Due to different AoDs for different users, we have 
\begin{equation}
  \left\{\begin{array}{l} 
    \mathbb{E} \left[ \bar{\mathbf{d}}_{i,m}^{H}\bar{\mathbf{d}}_{j,m} \right] = 1, i = j,\\
    \mathbb{E} \left[ \bar{\mathbf{d}}_{i,m}^{H}\bar{\mathbf{d}}_{j,m} \right] = 0, i \neq j.
    \end{array}\right.
\end{equation} 
And $\mathbb{E} \left[ \tilde{\mathbf{d}}_{i,m}^{H}\tilde{\mathbf{d}}_{j,m} \right]$ equals to $1$ when $i = j$, otherwise equals to $0$. When AoDs of the RIS and users at the BS are distinguishable, $\mathbb{E} \left[ \bar{\mathbf{d}}_{i,m}^{H} \bar{\mathbf{G}}_{m}\mathbf{\Phi}  \bar{\mathbf{h}}_{j,m} \right]$ and $\mathbb{E} \left[\bar{\mathbf{h}}_{i,m}^{H} \mathbf{\Phi}^{H} \bar{\mathbf{G}}_{m}^{H} \bar{\mathbf{d}}_{j,m}\right]$ are equal to $0$. For $\mathbb{E} \left[\bar{\mathbf{h}}_{i,m}^{H} \mathbf{\Phi}^{H} \bar{\mathbf{G}}_{m}^{H}\bar{\mathbf{G}}_{m}\mathbf{\Phi}  \bar{\mathbf{h}}_{j,m} \right]$, we have (\ref{hphig}), which is shown in the bottom of the next page.
\begin{figure*}[hb] 
  \hrulefill
  \centering
\begin{equation}
  \begin{aligned} \label{hphig}
  &\mathbb{E} \left[\bar{\mathbf{h}}_{i,m}^{H} \mathbf{\Phi}^{H} \bar{\mathbf{G}}_{m}^{H}\bar{\mathbf{G}}_{m}\mathbf{\Phi}  \bar{\mathbf{h}}_{j,m} \right]  \\
= &\mathbb{E} \left[ \mathbf{a}\left( \theta_{2,i,m}^{e}, \theta_{2,i,m}^{a}\right)\mathbf{\Phi}^{H} \mathbf{a}^{H}\left( \theta_{0,m}^{e}, \theta_{0,m}^{a}\right)\mathbf{b}\left( \phi_{0,m} \right) \mathbf{b}^{H}\left( \phi_{0,m} \right) \mathbf{a}\left( \theta_{0,m}^{e}, \theta_{0,m}^{a}\right)\mathbf{\Phi}\mathbf{a}^{H}\left( \theta_{2,k,m}^{e}, \theta_{2,j,m}^{a}\right)  \right] \\
= &\mathbb{E} \left[ \mathbf{a}\left( \theta_{2,i,m}^{e}, \theta_{2,i,m}^{a}\right)\mathbf{\Phi}^{H} \mathbf{a}^{H}\left( \theta_{0,m}^{e}, \theta_{0,m}^{a}\right) \mathbf{a}\left( \theta_{0,m}^{e}, \theta_{0,m}^{a}\right)\mathbf{\Phi}\mathbf{a}^{H}\left( \theta_{2,j,m}^{e}, \theta_{2,j,m}^{a}\right)  \right].
  \end{aligned}
\end{equation}
\end{figure*}

When $i=j$, the following items can be calculated as 
\begin{align}
  &\mathbb{E} \left[\tilde{\mathbf{h}}_{i,m}^{H} \mathbf{\Phi}^{H} \bar{\mathbf{G}}_{m}^{H}\bar{\mathbf{G}}_{m}\mathbf{\Phi}  \tilde{\mathbf{h}}_{j,m} \right] \notag \\ 
  = &\text{Tr}\left( \mathbf{\Phi}^{H} \bar{\mathbf{G}}_{m}^{H}\bar{\mathbf{G}}_{m}\mathbf{\Phi}\right)=\text{Tr}\left( \bar{\mathbf{G}}_{m}^{H}\bar{\mathbf{G}}_{m}\right)=1, \\ \label{res_tr}
  &\mathbb{E} \left[\bar{\mathbf{h}}_{i,m}^{H} \mathbf{\Phi}^{H} \tilde{\mathbf{G}}_{m}^{H}\tilde{\mathbf{G}}_{m}\mathbf{\Phi}  \bar{\mathbf{h}}_{j,m} \right]=\mathbb{E} \left[\bar{\mathbf{h}}_{i,m}^{H}\bar{\mathbf{h}}_{j,m} \right] = 1, \\
  &\mathbb{E} \left[\tilde{\mathbf{h}}_{i,m}^{H} \mathbf{\Phi}^{H} \tilde{\mathbf{G}}_{m}^{H}\tilde{\mathbf{G}}_{m}\mathbf{\Phi}  \tilde{\mathbf{h}}_{j,m} \right]= \mathbb{E} \left[\tilde{\mathbf{h}}_{i,m}^{H}\tilde{\mathbf{h}}_{j,m} \right] = 1,
\end{align}
where the result of (\ref{res_tr}) is obtained from \cite{Gan_RIS_2021}. For the case $i \neq j$, the above three items are equal to $0$.

Combining all above results, we have (\ref{res_hh}). The proof ends here.
\section{Proof of Lemma \ref{lemma3}}
Let $a_{1} = \sum_{t=1}^{T}-2 d_{0} d_{t} \cos \varphi_{t}$ and $a_{2} = \sum_{t=1}^{T}-2 d_{0} d_{t} \sin \varphi_{t}$, the problem becomes
\addtocounter{equation}{0}
\begin{IEEEeqnarray}{lcll} 
    &\min _{ \varphi_{0}} \quad & y_{3} = a_{1} \cos \varphi_{0} + a_{2}\sin \varphi_{0} \label{p_Ab} \\
  &\text{s.t.}  & \varphi_{0} \in [0,2\pi].     \IEEEyessubnumber\label{p_Ab_constraint1}
\end{IEEEeqnarray}
Define $x \triangleq\cos \varphi_{0}$, the first-order and second-order derivative are 
\begin{equation}
  \left\{\begin{array}{ll}
  \frac{\partial y_{3}}{\partial x} &= a_{1}- a_{2} \frac{x}{\sqrt{1-x^{2}}}, \\
  \frac{\partial^{2} y_{3}}{\partial x^2} &= -a_{2} \frac{1}{\sqrt{1-x^{2}}} - a_{2}x^2\left(1-x^{2}\right)^{-\frac{3}{2}}.
  \end{array}\right.
\end{equation}
There are two cases need to be discussed. Firstly, when $a_2 < 0$, $\frac{\partial^{2} y_{3}}{\partial x^2} > 0$, the first-order derivative is monotonically increasing. Due to $x \in [-1,1]$, the objective function decreases first and then increases. The optimal $\varphi_{0}$ is found by 
\begin{align} \label{varphi0_1}
  \varphi_{0}^{*}=\arccos \frac{a_{1}}{\sqrt{4 a_{2}^{2}+a_{1}^{2}}}.
\end{align}

Then, when $a_2 > 0$, the objective function increases first and then decreases. Then the optimal $\varphi_{0}$ is found by
\begin{equation} \label{varphi0_2}
  \varphi_{0}^{*} = 
  \left\{\begin{array}{l} 
      0, \quad y_{3}(x=1) \leq y_{3}(x=-1) ,\\
      \pi, \quad y_{3}(x=-1) \leq y_{3}(x=1). 
    \end{array}\right.
\end{equation} 
Combining (\ref{varphi0_1}) and (\ref{varphi0_2}), we have the conclusion of (\ref{varphi0_opt}).

\section{The RIS deployment optimaization when the BS employs MMSE precoder} \label{mmse-solution}
Firstly, the MMSE precoder can be expressed as $ \mathbf{U}_{m} = \left({\mathbf{H}}_{m}^{\text{eff}}\right)^{H}\left( {\mathbf{H}}_{m}^{\text{eff}}\left({\mathbf{H}}_{m}^{\text{eff}}\right)^{H} + \alpha_{m} \mathbf{I}\right)^{-1}$, where $\alpha_{m} = {\sum_{k=1}^{K}p_{k,m}}/{\sigma^{2}}$. The expression for user rate remains same as shown in (\ref{R_km}). With the MMSE precoder, the expression of $\| \mathbf{u}_{k,m}\|^{2}$ is given by $\| \mathbf{u}_{k,m}\|^{2} =\left[\left( {\mathbf{H}}_{m}^{\text{eff}}\left({\mathbf{H}}_{m}^{\text{eff}}\right)^{H} + \alpha_{m}\mathbf{I}\right)^{-1}\right]_{k,k}$. 

Then, according to the matrix inversion formula $\left( \mathbf{A} +\mathbf{B} \right)^{-1} = \mathbf{B}^{-1} \left( \mathbf{A}^{-1} +\mathbf{B}^{-1} \right)^{-1}\mathbf{A}^{-1}$, we have
\begin{equation}
  \begin{aligned}
    &\left({\mathbf{H}}_{m}^{\text{eff}}\left({\mathbf{H}}_{m}^{\text{eff}}\right)^{H} + \alpha_{m}\mathbf{I}\right)^{-1} \\
    = &\frac{1}{\alpha_{m}} \left(  \left({\mathbf{H}}_{m}^{\text{eff}}\left({\mathbf{H}}_{m}^{\text{eff}}\right)^{H} \right)^{-1}  + \frac{1}{\alpha_{m}} \right)^{-1} \left({\mathbf{H}}_{m}^{\text{eff}}\left({\mathbf{H}}_{m}^{\text{eff}}\right)^{H} \right)^{-1}.
    \end{aligned}
\end{equation}

Since $\mathbf{H}_{m}^{\text{eff}} \left( \mathbf{H}_{m}^{\text{eff}}\right)^{H}$ follows the Wishart distribution, based on the analysis in Section III and the properties of the Wishart distribution, we can obtain $\mathbb{E}\left[  \left({\mathbf{H}}_{m}^{\text{eff}}\left({\mathbf{H}}_{m}^{\text{eff}}\right)^{H} \right)^{-1} \right] = \frac{1}{N_{t}-K} \hat{\boldsymbol{\Sigma}}_{m}^{-1}$, and elements of $\hat{\boldsymbol{\Sigma}}_{m}^{-1}$ are given by
\begin{equation}
  \begin{aligned}
    \left[\hat{\boldsymbol{\Sigma}}_{m}^{-1} \right]_{i,j} = 
    \left\{
    \begin{array}{ll}
     \frac{1}{\kappa_{i}} - \frac{\tau \frac{|\Xi_{m,i}|^{2}}{\kappa_{i}^{2}}}{1 + \tau \sum_{k=1}^{K}\frac{|\Xi_{m,k}|^{2}}{\kappa_{k}}}, &i = j,\\
      - \frac{\tau \Xi_{m,i}\Xi_{m,j}\left[\tilde{\boldsymbol{\Sigma}}_{m}^{-1}\right]_{i,i}\left[\tilde{\boldsymbol{\Sigma}}_{m}^{-1}\right]_{j,j}}{1 + \tau \sum_{k=1}^{K}\frac{|\Xi_{m,k}|^{2}}{\kappa_{k}}}, &i \neq j.
    \end{array}
    \right.
    \end{aligned}
\end{equation}

Define $\mathbf{Z}_{m} \triangleq \left({\mathbf{H}}_{m}^{\text{eff}}\left({\mathbf{H}}_{m}^{\text{eff}}\right)^{H} \right)^{-1}  + \frac{1}{\alpha_{m}}$, then we have $\mathbf{Z}_{m}= \boldsymbol{\Lambda}_{m} + \frac{b}{N_{t}-K}\hat{\boldsymbol{\Xi}}_{m}\hat{\boldsymbol{\Xi}}_{m}^{H}$, where $\hat{\boldsymbol{\Xi}}_{m} =\tilde{\boldsymbol{\Sigma}}_{m}^{-1} \boldsymbol{\Xi}_{m}$.  $\boldsymbol{\Lambda}_{m}$ is a diagonal matrix whose the $(k,k)$-th element is ${\Lambda}_{m,k} = \frac{1}{N_{t}-K}\left[ \hat{\boldsymbol{\Sigma}}_{m}^{-1} \right]_{k,k} + \frac{1}{\alpha_{m}} - \frac{b}{N_{t}-K} |\hat{\Xi}_{m,k}|^{2}$, where $b =  - \frac{\tau}{1 + \tau \sum_{k=1}^{K}\frac{|\hat{\Xi}_{m,k}|^{2}}{\kappa_{k}}}$. By applying Sherman-Morrison formula, we have 
\begin{equation}
  \mathbf{Z}_{m}^{-1} = {\boldsymbol{\Lambda}}_{m}^{-1} - \frac{\hat{b}^{2}{\boldsymbol{\Lambda}}_{m}^{-1} \hat{\boldsymbol{\Xi}}_{m}\hat{\boldsymbol{\Xi}}_{m}^{H}{\boldsymbol{\Lambda}}_{m}^{-1}}{1 + \hat{b}^{2}\hat{\boldsymbol{\Xi}}_{m}^{H}{\boldsymbol{\Lambda}}_{m}^{-1}\hat{\boldsymbol{\Xi}}_{m}},
\end{equation}
where $\hat{b} = \frac{b}{N_{t}-K}$. Then the elements of $\mathbf{Z}_{m}^{-1}$ can be calculated as
\begin{equation}
  \begin{aligned}
    \left[\mathbf{Z}_{m}^{-1} \right]_{i,j} = 
    \left\{
    \begin{array}{ll}
     \frac{1}{\Lambda_{m,i}} - \frac{\hat{b}\frac{|\hat{\Xi}_{m,i}|^{2}}{\Lambda_{m,i}^{2}}}{1 + \hat{b} \sum_{k=1}^{K}\frac{|\hat{\Xi}_{m,k}|^{2}}{\Lambda_{m,k}}}, &i = j, \\
     - \frac{\hat{b} \hat{\Xi}_{m,i}\hat{\Xi}_{m,j}/(\Lambda_{m,i}\Lambda_{m,j})}{1 + \hat{b} \sum_{k=1}^{K}\frac{|\Xi_{m,k}|^{2}}{\Lambda_{m,k}}},& i \neq j.
    \end{array}
    \right.
    \end{aligned}
\end{equation}

Finally, the expression of $\left({\mathbf{H}}_{m}^{\text{eff}}\left({\mathbf{H}}_{m}^{\text{eff}}\right)^{H} + \alpha_{m}\mathbf{I}\right)^{-1} $ is given by $\frac{1}{\alpha_{m}\left(N_{t}-K\right)}\mathbf{Z}_{m}^{-1}\hat{\boldsymbol{\Sigma}}_{m}^{-1}$, and the $(k,k)$-th element is $\varpi_{m,k} =\frac{1}{\alpha_{m}\left(N_{t}-K\right)} \sum_{i=1}^{K}\left[\mathbf{Z}_{m}^{-1} \right]_{k,i}\left[ \hat{\boldsymbol{\Sigma}}_{m}^{-1} \right]_{i,k}$. The average user rate with MMSE precoder is 
\begin{equation}
  R_{k,m}^{\mathrm{MMSE} } = \log_{2}\left( 1 + \frac{p_{k,m}}{\sigma^{2} \varpi_{k,m}} \right).
\end{equation}

After obtaining the average user rate, the deployment optimization problem for the RIS location can still be solved using the heuristic algorithm proposed in Section IV.


\ifCLASSOPTIONcaptionsoff
  \newpage
\fi

\bibliographystyle{IEEEtran}
\bibliography{ref.bib}


\end{document}